\documentclass{article}

\usepackage{amsmath,amssymb,amsthm}
\usepackage{tikz}
\usepackage{tkz-graph}
\usepackage{tikz-cd}
\usepackage{tikz-qtree}
\usepackage{youngtab}
\usepackage[OT2,T1]{fontenc}
\DeclareSymbolFont{cyrletters}{OT2}{wncyr}{m}{n}

\DeclareMathSymbol{\Sha}{\mathalpha}{cyrletters}{"58}

\newtheorem{theorem}{Theorem}[section]
\newtheorem{lemma}{Lemma}[section]
\newtheorem{corollary}{Corollary}[section]

\newtheorem{definition}{Definition}[section]
\newtheorem{proposition}{Proposition}[section]
\newtheorem{remark}{Remark}[section]

\newtheorem{examples}{Example}[section]

\begin{document}
	
\bibliographystyle{abbrv}
	
\title{Explicit Representatives and Sizes of Cyclotomic Cosets and their Application to Cyclic Codes over Finite Fields}
	
\author{Li Zhu$^{1}$, Jinle Liu$^{2}$ and Hongfeng Wu$^{2}$\footnote{Corresponding author.}
\setcounter{footnote}{-1}
\footnote{E-Mail addresses:
 lizhumath@pku.edu.cn (L. Zhu), cohomologyliu@163.com (J. Liu), whfmath@gmail.com (H. Wu)}
\\
{1.~School of Mathematical Sciences, Guizhou Normal University, Guiyang, China}
\\
{2.~College of Science, North China University of technology, Beijing, China}}

\date{}
\maketitle

\thispagestyle{plain}
\setcounter{page}{1}	

\begin{abstract}
	Cyclotomic coset is a classical notion in the theory of finite field which has wide applications in various computation problems. Let $q$ be a prime power, and $n$ be a positive integer coprime to $q$. In this paper we determine explicitly the representatives and the sizes of all $q$-cyclotomic cosets modulo $n$ in the general settings. We introduce the definition of $2$-adic cyclotomic system, which is a profinite space consists of certain compatible sequences of cyclotomic cosets. A precise characterization of the structure of the $2$-adic cyclotomic system is given, which reveals the general formula for representatives of cyclotomic cosets. With the representatives and the sizes of $q$-cyclotomic cosets modulo $n$, we improve the formulas for the factorizations of $X^{n}-1$ and of $\Phi_{n}(X)$ over $\mathbb{F}_{q}$ given in \cite{Graner}. As a consequence, we classify the cyclic codes over finite fields via giving their generator polynomials. Moreover, the self-dual cyclic codes are determined and enumerated. \\


\end{abstract}

\section{Introduction}
Let $n$ be a positive integer, and $q$ be a prime power which is coprime to $n$. For each $\gamma \in \mathbb{Z}/n\mathbb{Z}$, the $q$-cyclotomic cosets modulo $n$ containing $\gamma$ is 
$$c_{n/q}(\gamma) = \{\gamma,\gamma q,\cdots,\gamma q^{\tau-1}\} \subseteq \mathbb{Z}/n\mathbb{Z},$$
where $\tau$ is the smallest integer such that $\gamma q^{\tau} \equiv \gamma \pmod{n}$. Any element in $c_{n/q}(\gamma)$ will be called a representative of $c_{n/q}(\gamma)$, the smallest among which, viewed as nonnegative integers less than $n$, will be called the leader of $c_{n/q}(\gamma)$. The order of the set $c_{n/q}(\gamma)$ will be called the size of $c_{n/q}(\gamma)$. We denote by $\mathrm{CR}_{n/q}$ the set of all $q$-cyclotomic cosets modulo $n$.

Cyclotomic coset is a basic notion which widely occurs in computation problems. Due to the fact that the $q$-cyclotomic cosets modulo $n$ completely determine the irreducible factorization of $X^{n}-1$ over $\mathbb{F}_{q}$, cyclotomic coset is a key ingredient in the theory of factorizations of polynomials over finite fields, and thus are needed for various problems in coding theory and crytography. A recent paper shows that cyclotomic cosets can be used to compute certain Gaussian sums, which has application in statistic physics. Such problems require a detailed description of cyclotomic cosets and precise computations of the associated parameters, such as the representatives, leaders, sizes and enumerations of the cosets.

Many results has been achieved under specific conditions, some of which we list as follows. Being restricted by the authors' knowledge, the list is limited. In a series of papers, for instance \cite{Chen}, \cite{Chen 2}, \cite{Sharma}, \cite{Liu}, \cite{Sharma 2}, \cite{Liu 2}, \cite{Wu 2}, etc., the representatives and sizes of the corresponding classes of cyclotomic cosets are determined in order to calculate certain constacyclic codes. In \cite{Chen 3} and \cite{Wang} the $q$-cyclotomic cosets contained in the subset $1+r\mathbb{Z}/nr\mathbb{Z}$ of $\mathbb{Z}/nr\mathbb{Z}$, where $\mathrm{gcd}(q,nr)=1$ and $r \mid q-1$, are characterized and enumerated. Applying the results on cyclotomic cosets, \cite{Chen 3} gives the enumeration of Euclidean self-dual codes, and \cite{Wang} construct several classes of $p^{h}$-LCD MDS codes. And in \cite{Yue} and \cite{Geraci}, concerning with stream cipher $m$-sequences and problems in statistic physics respectively, algorithms to calculate the leader of cyclotomic cosets are given.

The purpose of this paper is to give a characterization of the $q$-cyclotomic cosets modulo $n$, via determining an explicit full set of representatives of these cosets and calculating their sizes, that works in the general settings. The method is briefly sketched as follows. We first deal with the cases where the module $n$ is odd, and then move on to the general case. The case of odd $n$ is easier due to the following facts. Let $p_{1},\cdots,p_{s}$ be distinct odd primes not equaling to $p$, and $e_{1},\cdots,e_{s}$ be positive integers. Then there exists a $s$-tuple $(\eta_{1},\cdots,\eta_{s})$ of integers, called a primitive root system modulo $n = p_{1}^{e_{1}}\cdots p_{s}^{e_{s}}$, which satisfies that for $i = 1,\cdots,s$,
\begin{itemize}
	\item[(1)] $\eta_{i}$ is a primitive root modulo $p_{i}^{d}$ for all $d \geq 1$;
	\item[(2)] for any $j \neq i$, $\eta_{i} \equiv 1 \pmod{p_{j}^{e_{j}}}$.
\end{itemize}
A concrete construction of a primitive root system modulo $n$ is represented in Subsection \ref{sec 5}. A primitive root system $(\eta_{1},\cdots,\eta_{s})$ modulo $n$ gives rise to an explicit full set of representatives of the $q$-cyclotomic cosets modulo $n$, and as shown in the computation in Section \ref{sec 6}, this set of representatives inherits good properties from $(\eta_{1},\cdots,\eta_{s})$ to compute with. The sizes of all the $q$-cyclotomic cosets modulo $n$ are also determined.

Unfortunately, as there is no primitive root modulo $2^{d}$ for any $d \geq 3$, the form of the representatives modulo an odd integer cannot be generalized naively to the case where the module $n$ is even. Assume that $n = 2^{e_{0}}p_{1}^{e_{1}}\cdots p_{s}^{e_{s}}$, where $p_{1},\cdots,p_{s}$ are distinct odd primes different from $p$ and $e_{0},e_{1},\cdots,e_{s}$ are positive integers, and $n^{\prime} = p_{1}^{e_{1}}\cdots p_{s}^{e_{s}}$ is the maximal odd divisor of $n$. We show that the $q$-cyclotomic cosets modulo $n$ can be deduced from that modulo $n^{\prime}$ inductively as follow. For any positive integer $m$ that is coprime to $q$, if a $q$-cyclotomic coset $c_{m/q}(\gamma_{m})$, $\gamma_{m} \in \mathbb{Z}/m\mathbb{Z}$, is given, then the preimage $\pi_{2m/m}^{-1}(c_{m/q}(\gamma_{m}))$ along the projection $\pi_{2m/m}: \mathbb{Z}/2m\mathbb{Z} \rightarrow \mathbb{Z}/m\mathbb{Z}$ fits in one of the following conditions:
\begin{itemize}
	\item[(1)] If $v_{2}(\gamma_{m}) + v_{2}(q^{|c_{m/q}(\gamma_{m})|}-1) \geq v_{2}(m)+1$, then 
	$$\pi_{2m/m}^{-1}(c_{m/q}(\gamma_{m})) = c_{2m/q}(\gamma_{m}) \sqcup c_{2m/q}(\gamma_{m}+m);$$
	\item[(2)] If $v_{2}(\gamma_{m}) + v_{2}(q^{|c_{m/q}(\gamma_{m})|}-1) < v_{2}(m)+1$, then 
	$$\pi_{2m/m}^{-1}(c_{m/q}(\gamma_{m})) = c_{2m/q}(\gamma_{m}) = c_{2m/q}(\gamma_{m}+m).$$
\end{itemize}
This fact, together with the obtained results on cyclotomic cosets modulo an odd integer, gives rise to an inductive method to compute representatives and sizes of cyclotomic cosets in general. We exhibit an examples at the end of Section \ref{sec 1}.

To further depict how the $q$-cyclotomic cosets "grow" as the module is multiplied by $2$ successively, instead of the $q$-cyclotomic cosets modulo a fixed integer, we consider the projective system of the spaces of $q$-cyclotomic cosets
$$\cdots \rightarrow \mathcal{C}_{2^{2}n^{\prime}/q} \rightarrow \mathcal{C}_{2n^{\prime}/q} \rightarrow \mathcal{C}_{n^{\prime}/q}.$$
A key object to characterize is the profinite space $\varprojlim\limits_{N} \widehat{\pi}_{2^{N}n^{\prime}/n^{\prime}}^{-1}(c_{n^{\prime}/q}(\gamma))$ associated to any fixed $q$-cyclotomic coset $c_{n^{\prime}/q}(\gamma)$ modulo $n^{\prime}$, where $\widehat{\pi}_{2^{N}n^{\prime}/n^{\prime}}: \mathcal{C}_{2^{N}n^{\prime}/q} \rightarrow \mathcal{C}_{n^{\prime}/q}$ is induced from the canonical projection $\pi_{2^{N}n^{\prime}/n^{\prime}}: \mathbb{Z}/2^{N}n^{\prime}\mathbb{Z} \rightarrow \mathbb{Z}/n^{\prime}\mathbb{Z}$.
It turns out that this profinite space is equipped with rather neat structure. We prove that the elements in $\varprojlim\limits_{N} \widehat{\pi}_{2^{N}n^{\prime}/n^{\prime}}^{-1}(c_{n^{\prime}/q}(\gamma))$ can be fully determined by a specific family of $2$-adic power series $U_{d}(\gamma
)$, $d \in \mathbb{N} \cup \infty$. These formal power series $U_{d}(\gamma)$'s are called the generating series of the space $\varprojlim\limits_{N} \widehat{\pi}_{2^{N}n^{\prime}/n^{\prime}}^{-1}(c_{n^{\prime}/q}(\gamma))$. As a consequence, an explicit full set of representatives of the $q$-cyclotomic cosets modulo $n$ along with their sizes are given. The above contents constitute Section \ref{sec 1}. 

Section \ref{sec 6} is devoted to representing some applications of the results obtain in Section \ref{sec 1} on cyclic codes over finite fields. In \cite{Graner}, A. Graner gives general formulas for the irreducible factorizations of $X^{n}-1$ and of cyclotomic polynomials over $\mathbb{F}_{q}$. Based on the results in the last section and the multiple equal-difference structure of cyclotomic cosets, we give improved formula to factorize $X^{n}-1$ and $\Phi_{n}(X)$ over $\mathbb{F}_{q}$ in this section. The improvements are made in the following two aspects. With the full set of representatives and the sizes of $q$-cyclotomic cosets modulo $n$, the irreducible factors of $X^{n}-1$ and of $\Phi_{n}(X)$ can be determined precisely. And the characterization of multiple equal-difference representations of cyclotomic cosets leads to the extension fields of $\mathbb{F}_{q}$ with the lowest degrees over which the irreducible factors of $X^{n}-1$ over $\mathbb{F}_{q}$ can be further factorized into irreducible binomials, which simplify the computation to apply the formulas in practice. As an application of the irreducible factorization of $X^{n}-1$ over $\mathbb{F}_{q}$, we determine the cyclic codes over $\mathbb{F}_{q}$ via giving their generator polynomials. To deal with both the simple-rooted and the repeat-rooted cyclic codes, we allow the length $m$ of the codes to be divisible by $p$. Furthermore, we prove a criterion for a cyclic code to be self-dual, and enumerate the self-dual cyclic codes of length $m$ over $\mathbb{F}_{q}$.

\section{Preliminaries}
\subsection{Basic number theory}
Throughout this paper, unless stating otherwise, it is assumed that $p$ is a prime number, $q =p^{e}$ for a positive integer $e$, and $n$ is a positive integer that is coprime to $q$. If $n$ is factorized as $n = p_{1}^{e_{1}}\cdots p_{s}^{e_{s}}$, where $p_{1},\cdots,p_{s}$ are pairwise different prime numbers and $e_{1},\cdots,e_{s}$ are positive integers, the radical of $n$ is defined to be 
$$\mathrm{rad}(n) = p_{1}\cdots p_{s}.$$

If $m$ and $n$ are coprime integers, we denote by $\mathrm{ord}_{n}(m)$ the multiplicative order of $m$ in $(\mathbb{Z}/ n\mathbb{Z})^{\ast}$, i.e., the smallest positive integer such that
$$m^{\mathrm{ord}_{n}(m)} \equiv 1 \pmod{n}.$$

Let $\ell$ be a prime number. Denote by $v_{\ell}(n)$ the $\ell$-adic valuation of $n$, i.e., the maximal nonnegative integer such that $\ell^{v_{\ell}(n)} \mid n$. The following lift-the-exponent lemmas are well-known, which deal with the cases that $\ell$ is odd and that $\ell = 2$ respectively.

\begin{lemma}{\cite{Nezami}}\label{lem 3}
	Let $\ell$ be an odd prime number, and $m$ be an integer such that $\ell \mid m-1$. Then $v_{\ell}(m^{d}-1) = v_{\ell}(m-1) + v_{\ell}(d)$ for any positive integer $d$.
\end{lemma}

\begin{lemma}{\cite{Nezami}}\label{lem 2}
	Let $m$ be an odd integer, and $d$ be a positive integer.
	\begin{itemize}
		\item[(1)] If $m \equiv 1 \pmod{4}$, then
		$$v_{2}(m^{d}-1) = v_{2}(m-1) + v_{2}(d), \  v_{2}(m^{d}+1) = 1.$$
		\item[(2)] If $m \equiv 3 \pmod{4}$ and $d$ is odd, then
		$$v_{2}(m^{d}-1) = 1, \  v_{2}(m^{d}+1) = v_{2}(m+1).$$
		\item[(3)] If $m \equiv 3 \pmod{4}$ and $d$ is even, then
		$$v_{2}(m^{d}-1) = v_{2}(m+1) + v_{2}(d), \  v_{2}(m^{d}+1) = 1.$$
	\end{itemize}
\end{lemma}

If $\ell$ is an odd prime number, the multiplicative group $(\mathbb{Z}/ \ell^{d}\mathbb{Z})^{\ast}$, $d \geq 1$, is a cyclic group. Each generator of $(\mathbb{Z}/ \ell^{d}\mathbb{Z})^{\ast}$ is called a primitive root modulo $\ell^{d}$. In fact, we can find an integer $\eta$ which is a common primitive root module $\ell^{d}$ for all $d \geq 1$. The following lemma gives a criterion for such $\eta$.

\begin{lemma}{\cite{Hardy}}\label{lem 1}
	Let $\eta$ be a primitive root modulo an odd prime $\ell$. If $\eta^{\ell-1}$ is not congruent to $1$ modulo $\ell^{2}$, then $\eta$ is a primitive root modulo $\ell^{d}$ for all $d \geq 1$.
\end{lemma}

\subsection{Finite fields and cyclic codes over finite fields}\label{sec 2}
Let $\mathbb{F}_{q}$ be a finite field with $q$ elements, and $\mathbb{F}_{q}^{\ast}$ be the multiplicative group of the nonzero elements in $\mathbb{F}_{q}$. Let $\mathbb{F}_{q^{m}}/ \mathbb{F}_{q}$, $m \in \mathbb{N}^{+}$, be a finite field extension. The Galois group $\mathrm{Gal}(\mathbb{F}_{q^{m}}/ \mathbb{F}_{q})$ is cyclic, with a canonical generator called the Frobenious automorphism of $\mathbb{F}_{q^{m}}$ over $\mathbb{F}_{q}$:
$$\mathrm{Frob}_{q}: \mathbb{F}_{q^{m}} \rightarrow \mathbb{F}_{q^{m}}: \ x \mapsto x^{q}.$$

Let $\lambda \in \mathbb{F}_{q}^{\ast}$. There is a minimal positive integer $r$ such that $\lambda^{r} = 1$, which is called the order of $\lambda$ and is denoted by $\mathrm{ord}(\lambda)=r$. It is obvious that $\mathrm{ord}(\lambda) \mid q-1$. Let $f(X)$ be an irreducible polynomial over $\mathbb{F}_{q}$ that is not equal to $aX$ for any $a \neq 0$, the order $\mathrm{ord}(f)$ of $f(X)$ is defined to be the minimal positive integer $r$ such that $f(X) \mid X^{r}-1$. One can verify that $\mathrm{ord}(f) = \mathrm{ord}(\alpha)$ for any root $\alpha$ of $f(X)$ lying in some finite extension of $\mathbb{F}_{q}$.

For any positive integer $n$ coprime to $q$, there are exactly $n$ roots of $X^{n}-1$ lying in some finite extension of $\mathbb{F}_{q}$. A root $\zeta_{n}$ of $X^{n}-1$ which fails to be a root of $X^{m}-1$ for any $m < n$ is called a primitive $n$-th root of unity. In this paper, we fix a compatible family
$$\{\zeta_{n} | \ \mathrm{gcd}(n,q) = 1\}$$
of primitive roots of unity, where the compatibility means that for any $m$ and $n$, coprime to $q$, such that $m \mid n$, it holds that $\zeta_{n}^{\frac{n}{m}} = \zeta_{m}$.
	
Let $n$ be a positive integer such that $\mathrm{gcd}(n,q)=1$. The $n$-th cyclotomic polynomial $\Phi_{n}(X)$ is defined to be
$$\Phi_{n}(X) = \prod_{\substack{0 \leq i < n\\ \mathrm{gcd}(i,n)=1}}(X-\zeta_{n}^{i}).$$
It is clear form the definition that $X^{n}-1 = \prod\limits_{d \mid n}\Phi_{d}(X)$, and the M\"{o}bius inversion formula gives that
$$\Phi_{n}(X) = \prod_{d \mid n}(X^{d}-1)^{\mu(\frac{n}{d})},$$
where $\mu$ is the M\"{o}bius function.
	
Given a $\gamma \in \mathbb{Z}/ n\mathbb{Z}$, the $q$-cyclotomic coset modulo $n$ containing $\gamma$ is defined as
$$c_{n/q}(\gamma) = \{\gamma,\gamma q,\cdots,\gamma q^{\tau-1}\} \subseteq \mathbb{Z}/ n\mathbb{Z},$$
where $\tau$ is the smallest positive integer such that $\gamma q^{\tau} = \gamma$ in $\mathbb{Z}/n\mathbb{Z}$. Each element in $c_{n/q}(\gamma)$  is called a representative of the coset $c_{n/q}(\gamma)$. Regarding these elements as nonnegative integer less than $n$, the smallest one is called the leader of $c_{n/q}(\gamma)$. The order of $c_{n/q}(\gamma)$, which equals to $\tau$, is called the size of $c_{n/q}(\gamma)$. We denote by $\mathcal{C}_{n/q}$ the space of all $q$-cyclotomic cosets modulo $n$. In the following context we sometimes do not distinguish an element in $\mathbb{Z}/n\mathbb{Z}$ with its preimage in $\mathbb{Z}$ in computations, in which case the result should be understood to be modulo $n$.

The $q$-cyclotomic cosets modulo $n$ can be naturally identified with the orbits of the Frobenius automorphism $\mathrm{Frob}_{q}$ on the set of $n$-th roots of unity, thus they fully determine the irreducible factorizations of $X^{n}-1$ and of $\Phi_{n}(X)$. Concretely, to each $q$-cyclotomic coset $c_{n/q}(\gamma)$ modulo $n$ associating an irreducible polynomial
$$M_{\gamma,q}(X)= (X-\zeta_{n}^{\gamma})(X-\zeta_{n}^{\gamma q})\cdots(X-\zeta_{n}^{\gamma q^{\tau-1}})$$
over $\mathbb{F}_{q}$, then the irreducible factorizations of $X^{n}-1$ and of $\Phi_{n}(X)$ are given by
$$X^{n}-1 = \prod_{\gamma \in \mathcal{CR}_{n/q}}M_{\gamma,q}(X)$$
and
$$\Phi_{n}(X) = \prod_{\substack{\gamma \in \mathcal{CR}_{n/q}\\ \mathrm{gcd}(\gamma,n)=1}}M_{\gamma,q}(X)$$
respectively, where $\mathcal{CR}_{n/q}$ is any full set of representatives of $q$-cyclotomic cosets modulo $n$.

\subsection{A brief reminder on the multiple equal-difference representations of cyclotomic cosets}
In this subsection we recall the definition and some results of the multiple equal-difference representations of cyclotomic cosets. For a complete treatment the readers are referred to \cite{Zhu}.

\begin{definition}
	Let $E$ be a subset of $\mathbb{Z}/n\mathbb{Z}$. We say that $E$ is of equal difference if there is a positive integer $d \mid n$ such that $E$ can be represented as 
	$$E = \{\gamma,\gamma+d,\cdots,\gamma+(\frac{n}{d}-1)d\}.$$
	The integer $d$ is called the common difference of $E$. In particular, an equal-difference $q$-cyclotomic coset modulo $n$ is a $q$-cyclotomic coset modulo $n$ that is of equal difference.
\end{definition}

For any $\gamma \in \mathbb{Z}/n\mathbb{Z}$ we set $n_{\gamma} = \frac{n}{\mathrm{gcd}(\gamma,n)}$. It is trivial to see that for any $\gamma$ and $\gamma^{\prime}$ lying in the same $q$-cyclotomic coset modulo $n$ it holds that $n_{\gamma} = n_{\gamma^{\prime}}$. The criterion for a $q$-cyclotomic coset modulo $n$ being of equal difference is given below.

\begin{proposition}
	The cyclotomic coset $c_{n/q}(\gamma)$ is of equal difference if and only if the following two conditions are satisfied:
	\begin{description}
		\item[(\romannumeral1)] $\mathrm{rad}(n_{\gamma}) \mid q-1$;
		\item[(\romannumeral2)] $q \equiv 1 \pmod{4}$ if $8 \mid n_{\gamma}$.
	\end{description}
\end{proposition}

In general, a cyclotomic coset $c_{n/q}(\gamma)$ is not necessarily of equal difference, however it turns out that $c_{n/q}(\gamma)$ can always be expressed as a disjoint union of equal-difference subsets, which is called an equal-difference decomposition of $c_{n/q}(\gamma)$. If, moreover, the decomposition is in the form
$$c_{n/q}(\gamma) = \bigsqcup_{j=0}^{\mathrm{gcd}(t,\tau)-1}c_{n/q^{t}}(\gamma q^{j})$$
for some $t \in \mathbb{N}^{+}$, then it is called a multiple equal-difference representation of $c_{n/q}(\gamma)$.

Let $c_{n/q}(\gamma)$ be a $q$-cyclotomic coset modulo $n$, with $n_{\gamma} = \frac{n}{\mathrm{gcd}(\gamma,n)}$. Define a positive integer $\omega_{\gamma}$ by
\begin{equation*}
	\omega_{\gamma} = \left\{
	\begin{array}{lcl}
		2\mathrm{ord}_{\mathrm{rad}(n_{\gamma})}(q), \quad \mathrm{if} \ q^{\mathrm{ord}_{\mathrm{rad}(n_{\gamma})}(q)} \equiv 3 \pmod{4} \ \mathrm{and} \ 8 \mid n_{\gamma};\\
		\mathrm{ord}_{\mathrm{rad}(n_{\gamma})}(q), \quad \mathrm{otherwise}.
	\end{array} \right.
\end{equation*}
The following theorem determines all the multiple equal-difference representations of $c_{n/q}(\gamma)$.

\begin{theorem}
	All the distinct multiple equal-difference representations of $c_{n/q}(\gamma)$ are
	$$c_{n/q}(\gamma) = \bigsqcup_{j=0}^{t-1}c_{n/q^{t}}(\gamma q^{j}),$$
	where $t$ ranges over all positive divisors of $\tau$ that can be divided by $\omega_{\gamma}$.
\end{theorem}

\begin{corollary}\label{coro 2}
	The coarsest multiple equal-difference representation of $c_{n/q}(\gamma)$ is
	$$c_{n/q}(\gamma) = \bigsqcup_{j=0}^{\omega_{\gamma}-1}c_{n/q^{\omega_{\gamma}}}(\gamma q^{j}),$$
	while the finest multiple equal-difference representation of $c_{n/q}(\gamma)$ is is the trivial decomposition
	$$c_{n/q}(\gamma) = \bigsqcup_{j=0}^{\tau-1}\{\gamma q^{j}\}.$$
\end{corollary}

The multiple equal-difference representations of $q$-cyclotomic cosets modulo $n$ are particularly of interest to us as they encodes the irreducible factorizations of $X^{n}-1$ into binomials over finite extensions of $\mathbb{F}_{q}$. The explicit statements are given below.

\begin{theorem}
	\begin{description}
		\item[(1)] A $q$-cyclotomic coset $c_{n/q}(\gamma)$ modulo $n$ is of equal-difference if and only if the induced polynomial
		$$M_{\gamma,q}(X) = (X-\zeta_{n}^{\gamma})(X-\zeta_{n}^{\gamma q})\cdots(X-\zeta_{n}^{\gamma q^{\tau-1}})$$
		is a binomial.
		\item[(2)] For a positive integer $t$, the following statements are equivalent:
		\begin{description}
			\item[(\romannumeral1)] $\omega_{\gamma} \mid t$;
			\item[(\romannumeral2)] $c_{n/q}(\gamma) = \bigsqcup\limits_{j=0}^{\mathrm{gcd}(t,\tau)-1}c_{n/q^{t}}(\gamma q^{j})$ is a multiple equal-difference representation of $c_{n/q}(\gamma)$;
			\item[(\romannumeral3)] $M_{\gamma,q}(X)$ is factorized into irreducible binomials over $\mathbb{F}_{q^{t}}$.
		\end{description}
		\item[(3)] For a positive integer $t$, the following statements are equivalent:
		\begin{description}
			\item[(\romannumeral1)] $\omega_{1} \mid t$;
			\item[(\romannumeral2)] $c_{n/q}(\gamma) = \bigsqcup\limits_{j=0}^{\mathrm{gcd}(t,\tau)-1}c_{n/q^{t}}(\gamma q^{j})$ is a multiple equal-difference representation of $c_{n/q}(\gamma)$ for every $q$-cyclotomic coset $c_{n/q}(\gamma)$ modulo $n$;
			\item[(\romannumeral3)] $X^{n}-1$ is factorized into irreducible binomials over $\mathbb{F}_{q^{t}}$.
		\end{description}
	\end{description}
\end{theorem}

\begin{corollary}
	The field $\mathbb{F}_{q^{\omega_{1}}}$ is the smallest extension of $\mathbb{F}_{q}$ over which $X^{n}-1$ is factorized into irreducible binomials.
\end{corollary}

\subsection{Cyclic codes over finite fields}
A linear code of length $n$ over $\mathbb{F}_{q}$ is a $\mathbb{F}_{q}$-subspace of $\mathbb{F}_{q}^{n}$. Define the cyclic shift on $\mathbb{F}_{q}^{n}$ by
$$\tau(c_{0},c_{1},\cdots,c_{n-1}) = (c_{n-1},c_{0},\cdots,c_{n-2}).$$

\begin{definition}
	A linear code $\mathcal{C}$ of length $n$ is said to be cycic if $\tau(\mathcal{C}) = \mathcal{C}$.
\end{definition}

Since $\mathbb{F}_{q}^{n}$ is isomorphic to $\mathbb{F}_{q}[X]/ (X^{n}- 1)$ as $\mathbb{F}_{q}$-spaces, each word $c = (c_{0},c_{1},\cdots,c_{n-1})$ in a linear code $\mathcal{C}$ can be identified with a unique polynomial
$$c_{0} + c_{1}X + \cdots + c_{n-1}X^{n-1} \in \mathbb{F}_{q}[X]/ (X^{n}- 1),$$
which is called the polynomial representation of $c$. Under such identifications, the code $\mathcal{C}$ is cyclic if and only if it is an ideal of $\mathbb{F}_{q}[X]/ (X^{n}- 1)$. Thus the cyclic codes of length $n$ over $\mathbb{F}_{q}$ are in a natural one-to-one correspondence with the monic factors of $X^{n}-1$. Following the convention for polynomials, we say that a cyclic code $\mathcal{C}$ of length $n$ is simple-rooted if $p \nmid n$, while is repeated-rooted if $p \mid n$.

Let $x = (x_{0},x_{1},\cdots,x_{n-1})$ and $y = (y_{0},y_{1},\cdots,y_{n-1})$ be in $\mathbb{F}_{q}^{n}$. The Euclidean product of $x$ and $y$ is given by
$$\langle x, y\rangle = x_{0}y_{0} + x_{1}y_{1} + \cdots + x_{n-1}y_{n-1} \in \mathbb{F}_{q}.$$
Let $\mathcal{C}$ be a linear code of length $n$ over $\mathbb{F}_{q}$. The Euclidean dual code of $\mathcal{C}$ is defined as
$$\mathcal{C}^{\bot} = \{x \in \mathbb{F}_{q}^{n} : \ \langle y, x\rangle =0, \ \forall y \in \mathcal{C}\},$$
which is again a linear code of length $n$. The code $\mathcal{C}$ is called self-orthogonal if $\mathcal{C} \subseteq \mathcal{C}^{\bot}$, and is called self-dual if $\mathcal{C} = \mathcal{C}^{\bot}$. 

If $\mathcal{C} = (f(X))$ is a cyclic code of length $n$, where $f(X) \mid X^{n}-1$, one can verify that $\mathcal{C}^{\bot}$ is alaso a cyclic code of length $n$ which is generated by $g^{\ast}(X)$, where $g(X) = \frac{X^{n}-1}{f(X)}$ and $g^{\ast}(X)$ is the reciprocal polynomial of $g(X)$, given by $g(X)^{\ast} = g(0)^{-1}X^{\mathrm{deg}(g)}g(\frac{1}{X})$. 

\section{Explicit representatives and sizes of cyclotomic cosets}\label{sec 1}
Let $p$ be a prime number, $q=p^{e}$ be a power of $p$, and $n$ be a positive integer not divisible by $p$. This section is devoted to an explicit characterization of the $q$-cyclotomic cosets modulo $n$, via determining a precise full set of representatives and the sizes of these cosets. The first subsection serves as a preparation, which contains the definition and a concrete construction of a primitive root system. In Subsection \ref{sec 3} and \ref{sec 4} we deal with the $q$-cyclotomic cosets modulo $n$ in the cases where $n$ is odd and where $n$ is even respectively.

\subsection{Primitive root systems}\label{sec 5}
Let $p_{1},\cdots,p_{s}$ be pairwise different odd primes, and $e_{1},\cdots,e_{s}$ be positive integers. A $s$-tuple $\eta=(\eta_{1},\cdots,\eta_{s})$ of integers is said to be a primitive root system modulo $p_{1}^{e_{1}}\cdots p_{s}^{e_{s}}$ if for any $1 \leq i \leq s$,
\begin{itemize}
	\item[(1)] $\eta_{i}$ is a primitive root modulo $p_{i}^{d}$ for all $d \geq 1$;
	\item[(2)] $\eta_{i} \equiv 1 \pmod{p_{j}^{e_{j}}}$ for any $j \neq i$. 
\end{itemize}

\begin{proposition}\label{prop 2}
	There exists a primitive root system $\eta = (\eta_{1},\cdots,\eta_{s})$ modulo $p_{1}^{e_{1}}\cdots p_{s}^{e_{s}}$.
\end{proposition}

\begin{proof}
	For any $1 \leq i \leq s$, we begin with a random primitive root $\mu_{i}^{\prime}$ modulo $p_{i}$. If $p_{i}^{2} \nmid \mu_{i}^{\prime p_{i}-1}-1$ we set $\mu_{i} = \mu_{i}^{\prime}$; otherwise we set $\mu_{i} = \mu_{i}^{\prime} + p_{i}$. Then $\mu_{i}$ is a primitive root modulo $p_{i}$ such that
	$$\mu_{i}^{p_{i}-1} \not\equiv 1 \pmod{p_{i}^{2}}.$$
	It follows from Lemma \ref{lem 1} that $\mu_{i}$ is a primitive root modulo $p_{i}^{d}$ for all $d \geq 1$. 
	
	Let $\eta_{i} = \mu_{i} + (1-\mu_{i})p_{i}^{\phi(p_{1}^{e_{1}})\cdots\widehat{\phi(p_{i}^{e_{i}})}\cdots \phi(p_{s}^{e_{s}})}$, where $\phi(p_{1}^{e_{1}})\cdots\widehat{\phi(p_{i}^{e_{i}})}\cdots \phi(p_{s}^{e_{s}})$ denotes the product of all $\phi(p_{j}^{e_{j}})$ for $j \neq i$. Then we have
	$$\eta_{i}^{p_{i}-1} = (\mu_{i} + (1-\mu_{i})p_{i}^{\phi(p_{1}^{e_{1}})\cdots\widehat{\phi(p_{i}^{e_{i}})}\cdots \phi(p_{s}^{e_{s}})})^{p_{i}-1} \equiv \mu_{i}^{p_{i}-1} \not\equiv 1 \pmod{p_{i}^{2}}$$
	and
	$$\eta_{i} \equiv \mu_{i} + (1- \mu_{i}) = 1 \pmod{p_{j}^{e_{j}}}, \ j \neq i.$$
	The conclusion follows immediately. 
\end{proof}

Let
$$\widehat{\mathbb{Z}}^{\prime} = \varprojlim_{N \ \mathrm{odd}} \mathbb{Z}/ N\mathbb{Z} = \prod_{\ell \ \mathrm{odd} \ \mathrm{prime}}\mathbb{Z}_{\ell}$$
be the quotient of the profinite completion $\widehat{\mathbb{Z}}$ of $\mathbb{Z}$ modulo $\mathbb{Z}_{2}$. For any odd prime $\ell_{0}$, there is a common primitive root $\eta_{\ell_{0}}$ modulo $\ell_{0}^{d}$ for all $d \geq 1$. Define an element $\widehat{\eta}_{\ell_{0}}$ in $\widehat{\mathbb{Z}}^{\prime}$ by setting the component of $\widehat{\eta}_{\ell_{0}}$ modulo every $\mathbb{Z}/\ell_{0}^{d}\mathbb{Z}$ to be $\eta_{\ell_{0}}$, while the component modulo $\mathbb{Z}/\ell\mathbb{Z}$ for $\ell \neq \ell_{0}$ is $1$. Such an element is often convenient to compute with, however it does not lie in $\mathbb{Z}$. The proposition below says that a primitive root system actually comes from projections of elements in this form.

\begin{proposition}
	Let $p_{1},\cdots,p_{s}$ be distinct odd primes, and $e_{1},\cdots,e_{s}$ be positive integers. Then a primitive root systems modulo $p_{1}^{e_{1}}\cdots p_{s}^{e_{s}}$ is exactly a lifting of a $s$-tuple 
	$$(\mathrm{pr}(\widehat{\eta}_{p_{1}}),\cdots,\mathrm{pr}(\widehat{\eta}_{p_{s}})) \in (\mathbb{Z}/p_{1}^{e_{1}}\cdots p_{s}^{e_{s}}\mathbb{Z})^{s},$$
	to $\mathbb{Z}$, where $\mathrm{pr}$ denotes the natural projection from $\widehat{\mathbb{Z}}^{\prime}$ onto $\mathbb{Z}/p_{1}^{e_{1}}\cdots p_{s}^{e_{s}}\mathbb{Z}$, and $\widehat{\eta}_{p_{i}}$, $1 \leq i \leq s$, are the elements in $\widehat{\mathbb{Z}}^{\prime}$ induced by some primitive roots $\eta_{i}$ modulo $p_{i}^{d}$ for $d \geq 1$ in the above way.
\end{proposition}
	
\subsection{$q$-cyclotomic cosets modulo an odd integer}\label{sec 3}
In this subsection we assume that $n = p_{1}^{e_{1}}\cdots p_{s}^{e_{s}}$ is an odd positive integer, where $p_{1},\cdots,p_{s}$ are distinct odd primes different from $p$ and $e_{1},\cdots,e_{s}$ are positive integers. For any integers $y_{1},\cdots,y_{s}$ such that $0 \leq y_{i} \leq e_{i}$, $i=1,\cdots,s$, we set the following notations:
\begin{itemize}
    \item[(1)] $f_{i,y_{i}} = \mathrm{ord}_{p_{i}^{e_{i}-y_{i}}}(q)$;
	\item[(2)] $g_{i,y_{i}} = \phi(p_{i}^{e_{i}-y_{i}})/ f_{i,y_{i}}$; and
	\item[(3)] $\tau_{y_{1}\cdots y_{s}} = \mathrm{ord}_{p_{1}^{e_{1}-y_{1}}\cdots p_{s}^{e_{s}-y_{s}}}(q) = \mathrm{lcm}(f_{1,y_{1}},\cdots,f_{s,y_{s}})$.
\end{itemize}
	
\begin{theorem}\label{thm 3}
	Fix a system of primitive roots $\eta=(\eta_{1},\cdots,\eta_{s})$ modulo $n = p_{1}^{e_{1}}\cdots p_{s}^{e_{s}}$. Then all the distinct $q$-cyclotomoic cosets modulo $n$ are given by
	$$c_{n/q}(\eta_{1}^{x_{1}}\cdots\eta_{s}^{x_{s}}p_{1}^{y_{1}}\cdots p_{s}^{y_{s}}) = \{\eta_{1}^{x_{1}}\cdots\eta_{s}^{x_{s}}p_{1}^{y_{1}}\cdots p_{s}^{y_{s}}, \eta_{1}^{x_{1}}\cdots\eta_{s}^{x_{s}}p_{1}^{y_{1}}\cdots p_{s}^{y_{s}}q,\cdots, \eta_{1}^{x_{1}}\cdots\eta_{s}^{x_{s}}p_{1}^{y_{1}}\cdots p_{s}^{y_{s}}q^{\tau_{y_{1}\cdots y_{s}}-1}\}$$
	for 
	\begin{description}
		\item[(1)] $0 \leq y_{i} \leq e_{i}$, $i = 1,\cdots,s$;
		\item[(2a)] $0 \leq x_{1} \leq g_{1,y_{1}}-1$; and 
		\item[(2b)] $0 \leq x_{i} \leq g_{i,y_{i}}\cdot \mathrm{gcd}(\mathrm{lcm}(f_{1,y_{1}},\cdots,f_{i-1,y_{i-1}}),f_{i,y_{i}}) - 1$, $i = 2,\cdots,s$.
	\end{description}
\end{theorem}

\begin{proof}
	First we prove that the given $q$-cyclotomic cosets are all distinct. Assume that
	$$c_{n/q}(\eta_{1}^{x_{1}}\cdots\eta_{s}^{x_{s}}p_{1}^{y_{1}}\cdots p_{s}^{y_{s}}) = c_{n/q}(\eta_{1}^{x_{1}^{\prime}}\cdots\eta_{s}^{x_{s}^{\prime}}p_{1}^{y_{1}^{\prime}}\cdots p_{s}^{y_{s}^{\prime}})$$ 
	where $x_{i}$ and $y_{i}$ (resp. $x_{i}^{\prime}$ and $y_{i}^{\prime}$), $i=1,\cdots,s$, satisfy the conditions $(1)$, $(2a)$ and $(2b)$. Then there exists a positive integer $d$ such that 
	\begin{equation}\label{eq 3}
		\eta_{1}^{x_{1}^{\prime}}\cdots\eta_{s}^{x_{s}^{\prime}}p_{1}^{y_{1}^{\prime}}\cdots p_{s}^{y_{s}^{\prime}} \equiv \eta_{1}^{x_{1}}\cdots\eta_{s}^{x_{s}}p_{1}^{y_{1}}\cdots p_{s}^{y_{s}}q^{d} \pmod{p_{1}^{e_{1}}\cdots p_{s}^{e_{s}}}.
	\end{equation}
	Since $\eta_{1},\cdots \eta_{s}$ and $q$ are coprime to $n$, comparing the greatest common divisor of the LHS and $n$ with that of the RHS and $n$ yields that $y_{1} = y_{1}^{\prime}, \cdots, y_{s} = y_{s}^{\prime}$. Equation \eqref{eq 3} is then reduced to 
	\begin{equation}\label{eq 4}
		\eta_{1}^{x_{1}^{\prime}}\cdots\eta_{s}^{x_{s}^{\prime}} \equiv \eta_{1}^{x_{1}}\cdots\eta_{s}^{x_{s}}q^{d} \pmod{p_{1}^{e_{1}-y_{1}}\cdots p_{s}^{e_{s}-y_{s}}}.
	\end{equation}
	Remembering that for each $1 \leq i \leq s$, $\eta_{i} \equiv 1 \pmod{p_{j}^{e_{j}}}$ for any $j \neq i$, thus we have
	\begin{equation}
		\left\{
		\begin{array}{rcl}
			\eta_{1}^{x_{1}^{\prime}-x_{1}} \equiv q^{d} \pmod{p_{1}^{e_{1}-y_{1}}}\\
			\eta_{2}^{x_{2}^{\prime}-x_{2}} \equiv q^{d} \pmod{p_{2}^{e_{2}-y_{2}}}\\
			\cdots \cdots\\
			\eta_{s}^{x_{s}^{\prime}-x_{s}} \equiv q^{d} \pmod{p_{s}^{e_{s}-y_{s}}}
		\end{array} \right.
	\end{equation}
	We show that $x_{i} = x_{i}^{\prime}$, $1 \leq i \leq s$, by induction. The congruence $\eta_{1}^{x_{1}^{\prime}-x_{1}} \equiv q^{d} \pmod{p_{1}^{e_{1}-y_{1}}}$ implies that
	$$\eta_{1}^{(x_{1}^{\prime}-x_{1})f_{1,y_{1}}} \equiv q^{d\cdot f_{1,y_{1}}} \equiv 1 \pmod{p_{1}^{e_{1}-y_{1}}},$$
	and therefore $\phi(p_{1}^{e_{1}-y_{1}}) \mid (x_{1}^{\prime}-x_{1})f_{1,y_{1}}$. Since $0 \leq x_{1}, x_{1}^{\prime} \leq g_{1,y_{1}} - 1$, it holds $x_{1} = x_{1}^{\prime}$. Supposes that $x_{1} = x_{1}^{\prime}, \cdots, x_{i-1} = x_{i-1}^{\prime}$ for $2 \leq i \leq s$. Then 
	$$q^{d} \equiv 1 \pmod{p_{1}^{e_{1}-y_{1}}\cdots p_{i-1}^{e_{i-1}-y_{i-1}}},$$
	which indicates that $\mathrm{lcm}(f_{1,y_{1}},\cdots,f_{i-1,y_{i-1}}) \mid d$. As $\eta_{i}^{x_{i}^{\prime}-x_{i}} \equiv q^{d} \pmod{p_{i}^{e_{i}-y_{i}}}$, then
	$$\eta_{i}^{(x_{i}^{\prime}-x_{i})\cdot \frac{f_{i,y_{i}}}{\mathrm{gcd}(\mathrm{lcm}(f_{1,y_{1}},\cdots,f_{i-1,y_{i-1}}),f_{i,y_{i}})}} \equiv q^{d\cdot \frac{f_{i,y_{i}}}{\mathrm{gcd}(\mathrm{lcm}(f_{1,y_{1}},\cdots,f_{i-1,y_{i-1}}),f_{i,y_{i}})}} \equiv 1 \pmod{p_{i}^{e_{i}-y_{i}}},$$
	and therefore 
	$$\phi(p_{i}^{e_{i}-y_{i}}) \mid (x_{i}^{\prime}-x_{i})\cdot \frac{f_{i,y_{i}}}{\mathrm{gcd}(\mathrm{lcm}(f_{1,y_{1}},\cdots,f_{i-1,y_{i-1}}),f_{i,y_{i}})}.$$ 
	It follows from $0 \leq x_{i}, x_{i}^{\prime} \leq g_{i,y_{i}} \cdot \mathrm{gcd}(\mathrm{lcm}(f_{1,y_{1}},\cdots,f_{i-1,y_{i-1}}),f_{i,y_{i}})-1$ that $x_{i} = x_{i}^{\prime}$.
	
	On the other hand, there are in total
	\begin{footnotesize}
		\begin{align*}
			&\sum_{0 \leq y_{1} \leq e_{1}}\cdots \sum_{0 \leq y_{s} \leq e_{s}} \dfrac{\phi(p_{1}^{e_{1}-y_{1}})}{f_{1,y_{1}}} \dfrac{\phi(p_{2}^{e_{2}-y_{2}})}{f_{2,y_{2}}} \mathrm{gcd}(f_{1,y_{1}},f_{2,y_{2}}) \cdots  \dfrac{\phi(p_{s}^{e_{s}-y_{s}})}{f_{s,y_{s}}} \mathrm{gcd}(\mathrm{lcm}(f_{1,y_{1}},\cdots,f_{s-1,y_{s-1}}),f_{s,y_{s}})\cdot \mathrm{lcm}(f_{1,y_{1}},\cdots,f_{s,y_{s}})\\
			&= \sum_{0 \leq y_{1} \leq e_{1}}\cdots \sum_{0 \leq y_{s} \leq e_{s}} \phi(p_{1}^{e_{1}-y_{1}}) \cdot\cdots\cdot \phi(p_{s}^{e_{s}-y_{s}}) = p_{1}^{e_{1}}\cdots p_{s}^{e_{s}}
		\end{align*}
	\end{footnotesize}
	elements contained in the cosets given in Theorem \ref{thm 3}. It follows that they are exactly all the distinct $q$-cyclotomic cosets modulo $n = p_{1}^{e_{1}}\cdots p_{s}^{e_{s}}$.
\end{proof}

\begin{remark}
	In the following context, we will write $(2a)$ and $(2b)$ uniformly as 
	\begin{description}
	    \item[(2)] $0 \leq x_{i} \leq g_{i,y_{i}}\cdot \mathrm{gcd}(\mathrm{lcm}(f_{1,y_{1}},\cdots,f_{i-1,y_{i-1}}),f_{i,y_{i}}) - 1$, $i =1,2,\cdots,s$
	\end{description}
	For $i = 1$, the term $\mathrm{gcd}(\mathrm{lcm}(f_{1,y_{1}},\cdots,f_{i-1,y_{i-1}}),f_{i,y_{i}})$ should be understood to be $1$. We call a $2s$-tuple $(y,x) = (y_{1},\cdots,y_{s},x_{1},\cdots,x_{s})$ which satisfies $(1)$ and $(2)$ in Theorem \ref{thm 3} an applicable tuple with respect to $n$ and $q$, and denote by $\Sigma = \Sigma(n;q)$ the set of all applicable tuples with respect to $n$ and $q$. Given a system of primitive roots $\eta=(\eta_{1},\cdots,\eta_{s})$ modulo $n$, we will simply write
	$$\eta^{(y,x)} = \eta_{1}^{x_{1}}\cdots\eta_{s}^{x_{s}}p_{1}^{y_{1}}\cdots p_{s}^{y_{s}}$$
	and
	$$\tau_{y} = \tau_{y_{1} \cdots y_{s}}$$
	for any $(y,x) = (y_{1},\cdots,y_{s},x_{1},\cdots,x_{s}) \in \Sigma$.
\end{remark}

\begin{corollary}
	The set
	$$\mathcal{CR}_{n/q} = \{\eta^{(y,x)} | (y,s) \in \Sigma\}$$
	is a full set of representatives of $q$-cyclotomic cosets modulo $n$.
\end{corollary}

On the other hand, the sizes of $q$-cyclotomic cosets modulo $n$ can be given more precisely as below.

\begin{proposition}\label{prop 4}
	For any $(y,x)\in \Sigma$, the coset $c_{n/q}(\eta^{(y,x)})$ has size
	$$\tau_{y}= \omega_{y}\cdot p_{1}^{M_{y_{1}}}\cdots p_{s}^{M_{y_{s}}},$$
	where for $i = 1,\cdots,s$, $m_{y_{i}} = \mathrm{min}\{e_{i}-y_{i},1\}$, $\omega_{y}=\mathrm{lcm}(\mathrm{ord}_{p_{1}^{m_{y_{1}}}}(q),\cdots,\mathrm{ord}_{p_{s}^{m_{y_{s}}}}(q))$ and $M_{y_{i}} = \mathrm{max}\{e_{i}-y_{i}-v_{p_{i}}(q^{\omega_{y}}-1),0\}$.
\end{proposition}

\begin{proof}
	Notice that $ p_{1}^{m_{y_{1}}}\cdots p_{s}^{m_{y_{s}}} = \mathrm{rad}(p_{1}^{e_{1}-y_{1}}\cdots p_{s}^{e_{s}-y_{s}})$, therefore
	$$\omega_{y} = \mathrm{lcm}(\mathrm{ord}_{p_{1}^{m_{y_{1}}}}(q),\cdots,\mathrm{ord}_{p_{s}^{m_{y_{s}}}}(q)) \mid \mathrm{lcm}(\mathrm{ord}_{p_{1}^{e_{1}-y_{1}}}(q),\cdots,\mathrm{ord}_{p_{s}^{e_{s}-y_{s}}}(q)) = \tau_{y}.$$
	Applying Lemma \ref{lem 2} for the primes $p_{1},\cdots,p_{s}$, respectively, yields that 
	$$\tau_{y} = \omega_{y}\cdot p_{1}^{M_{y_{1}}}\cdots p_{s}^{M_{y_{s}}}.$$
\end{proof}

\subsection{$q$-cyclotomic cosets modulo an even positive integer}\label{sec 4}
\subsubsection{The classification of $q$-cyclotomic cosets modulo $n$}
Let $m$ and $n$ be positive integers such that $m \mid n$. There is a canonical projection
$$\pi_{n/m}: \mathbb{Z}/n\mathbb{Z} \rightarrow \mathbb{Z}/m\mathbb{Z}; \ a+n\mathbb{Z} \mapsto a+m\mathbb{Z}.$$
If, moreover, both $m$ and $n$ are coprime to $q$, then the image $\pi_{n/m}(c_{n/q}(\gamma))$ of any $q$-cyclotomic coset $c_{n/q}(\gamma)$ modulo $n$ is exactly the coset $c_{m/q}(\pi_{n/m}(\gamma))$ modulo $m$ containing $\pi_{n/m}(\gamma)$. Consequently, $\pi_{n/m}$ induces a map
$$\widehat{\pi}_{n/m}: \mathcal{C}_{n/q} \rightarrow \mathcal{C}_{m/q}: \ c_{n/q}(\gamma) \mapsto c_{m/q}(\pi_{n/m}(\gamma)).$$
It is obvious to see that $\widehat{\pi}_{n/m}$ is a surjection.

Now let $n=2^{e_{0}}p_{1}^{e_{1}}\cdots p_{s}^{e_{s}}$, where $p_{1},\cdots,p_{s}$ are distinct odd primes different from $p$ and $e_{0},e_{1},\cdots,e_{s}$ are positive integers, and let $n^{\prime} = p_{1}^{e_{1}}\cdots p_{s}^{e_{s}}$ be the maximal odd divisor of $n$. We show that the space $\mathcal{C}_{n/q}$ can be naturally classified by the preimages of the cosets in $\mathcal{C}_{n^{\prime}/q}$ along $\pi_{n/n^{\prime}}$.

\begin{proposition}\label{prop 1}
	Let $\eta=(\eta_{1},\cdots,\eta_{s})$ be a primitive root system modulo $n^{\prime}$. Then the space $\mathcal{C}_{n/q}$ of $q$-cyclotomic cosets modulo $n$ has the following partition:
	\begin{equation}\label{eq 1}
		\mathcal{C}_{n/q} = \bigsqcup_{(y,x) \in \Sigma(n^{\prime};q)} \widehat{\pi}_{n/n^{\prime}}^{-1}(c_{n^{\prime}/q}(\eta^{(y,x)})).
	\end{equation}
	Moreover, the partition \eqref{eq 1} is compatible with the projective system
	$$\cdots \rightarrow \mathcal{C}_{2^{2}n^{\prime}/q} \rightarrow \mathcal{C}_{2n^{\prime}/q} \rightarrow \mathcal{C}_{n^{\prime}/q}$$
	in the sense that for any $N = 2^{E_{0}}p_{1}^{e_{1}}\cdots p_{s}^{e_{s}}$ with $E_{0} \geq e_{0}$ it holds that
	$$\widehat{\pi}_{N/n^{\prime}}^{-1}(c_{n^{\prime}/q}(\eta^{(y,x)})) = \bigsqcup_{c_{n/q}(\gamma) \in \widehat{\pi}_{n/n^{\prime}}^{-1}(c_{n^{\prime}/q}(\eta^{(y,x)}))} \widehat{\pi}_{N/n}^{-1}(c_{n/q}(\gamma))$$
	for any $(y,x) \in \Sigma(n^{\prime};q)$, and 
	$$\mathcal{C}_{N/q} = \bigsqcup_{(y,x) \in \Sigma(n^{\prime};q)}\left(\bigsqcup_{c_{n/q}(\gamma) \in \widehat{\pi}_{n/n^{\prime}}^{-1}(c_{n^{\prime}/q}(\eta^{(y,x)}))} \widehat{\pi}_{N/n}^{-1}(c_{n/q}(\gamma))\right).$$
\end{proposition}

\begin{proof}
	The first assertion follows from Theorem \ref{thm 3} and the fact that $\widehat{\pi}_{n/n^{\prime}}$ is surjective. If $N = 2^{E_{0}}p_{1}^{e_{1}}\cdots p_{s}^{e_{s}}$ with $E_{0} \geq e_{0}$, then the projection $\pi_{N/n^{\prime}}$ is the composition $\pi_{n/n^{\prime}}\cdot\pi_{N/n}$, which implies that 
	$$\widehat{\pi}_{N/n^{\prime}} = \widehat{\pi}_{n/n^{\prime}}\cdot \widehat{\pi}_{N/n}.$$
	Therefore the last conclusion holds.
\end{proof}

\subsubsection{Cyclotomic system}
According to Proposition \ref{prop 1}, to determine a full set or representatives of $q$-cyclotomic cosets modulo $n$ it suffices to compute the preimage $\widehat{\pi}_{n/n^{\prime}}^{-1}(c_{n^{\prime}/q}(\eta^{(y,x)}))$ for all $(y,x) \in \Sigma(n^{\prime};q)$. For this purpose we introduce the definition of cyclotomic system in this subsection. 

Let $q$ be a power of an odd prime $p$, and $n= p_{1}^{e_{1}}\cdots p_{s}^{e_{s}}$ be an odd integer, where $p_{1},\cdots,p_{s}$ are distinct odd primes different from $p$ and $e_{1},\cdots,e_{s}$ are positive integers. Consider the chain of projections
$$\cdots \rightarrow \mathbb{Z}/2^{2}n\mathbb{Z} \rightarrow \mathbb{Z}/2n\mathbb{Z} \rightarrow \mathbb{Z}/n\mathbb{Z}.$$
It gives rise to a chain of maps
$$\cdots \rightarrow \mathcal{C}_{2^{2}n/q} \rightarrow \mathcal{C}_{2n/q} \rightarrow \mathcal{C}_{n/q},$$
which forms a projective system of finite sets. Moreover, for any $q$-cyclotomic coset $c_{n/q}(\gamma)$ modulo $n$, the subsets
$$\cdots \rightarrow \widehat{\pi}_{2^{2}n/n}^{-1}(c_{n/q}(\gamma)) \rightarrow \widehat{\pi}_{2n/n}^{-1}(c_{n/q}(\gamma)) \rightarrow \{c_{n/q}(\gamma)\}$$
also forms a projective system.

\begin{definition}
	Define the ($2$-adic) q-cyclotomic system with base module $n$ to be the projective limit
	$$\mathcal{PC}_{n/q} = \varprojlim_{i \in \mathbb{N}}\mathcal{C}_{2^{i}n/q}.$$
	Fix a $q$-cyclotomic coset $c_{n/q}(\gamma)$ modulo $n$, the ($2$-adic) $q$-cyclotomic system over $c_{n/q}(\gamma)$ is defined to be the projective limit 
	$$\mathcal{PC}_{n/q}(\gamma) = \varprojlim_{i \in \mathbb{N}}\widehat{\pi}_{2^{i}n/n}^{-1}(c_{n/q}(\gamma)).$$
\end{definition}

Explicitly, elements in the profinite space $\mathcal{PC}_{n/q}$ are compatible sequences $(c_{2^{i}n/q}(\gamma_{i}))_{i \in \mathbb{N}}$ of cyclotomic cosets, where the compatibility means that for any $i_{1} \geq i_{2}$ it holds that 
$$\pi_{2^{i_{1}}n/2^{i_{2}}n}(c_{2^{i_{1}}n/q}(\gamma_{i_{1}})) = c_{2^{i_{2}}n/q}(\gamma_{i_{2}}).$$
And $\mathcal{PC}_{n/q}(\gamma)$ is the sub-profinite space of $\mathcal{PC}_{n/q}$ consisting of the sequences with the first component being $c_{n/q}(\gamma)$. It follows from the definition that $\mathcal{PC}_{n/q}$ can be decomposed as 
$$\mathcal{PC}_{n/q} = \bigsqcup_{\gamma \in \mathcal{CR}_{n/q}} \mathcal{PC}_{n/q}(\gamma),$$
where $\mathcal{CR}_{n/q}$ is any full set of representatives of $q$-cyclotomic cosets modulo $n$.

\begin{lemma}\label{lem 4}
	Let $m$ be an arbitrary positive integer, and $q$ be an odd prime power that is coprime to $m$. Let $\gamma$ be an element in $\mathbb{Z}/ m\mathbb{Z}$ with the associated $q$-cyclotomic coset modulo $m$ given by
	$$c_{m/q}(\gamma) = \{\gamma,\gamma q,\cdots,\gamma q^{\tau-1}\}.$$
	\begin{itemize}
		\item[(1)] If $2m \mid \gamma q^{\tau}-\gamma$, that is, $v_{2}(q^{\tau}-1) + v_{2}(\gamma) \geq v_{2}(m)+1$, then viewing $\gamma$ as an element in $\mathbb{Z}/2m\mathbb{Z}$, the $q$-cyclotomic coset modulo $2m$ containing $\gamma$ is 
		$$c_{2m/q}(\gamma) =  \{\gamma,\gamma q,\cdots,\gamma q^{\tau-1}\}.$$
		Moreover, the $q$-cyclotomic coset
		$$c_{2m/q}(m+\gamma) = \{m+\gamma,(m+\gamma)q,\cdots,(m+\gamma)q^{\tau-1}\}$$
		 modulo $2m$ containing $m+\gamma$ is disjoint with $c_{2m/q}(\gamma)$, and the union $c_{2m/q}(\gamma) \sqcup c_{2m/q}(m+\gamma)$ is exactly the preimage of $c_{m/q}(\gamma)$ under the projection $\pi_{2m/m}$.
		\item[(2)] If $2m \nmid \gamma q^{\tau}-\gamma$, that is, $v_{2}(q^{\tau}-1) + v_{2}(\gamma) < v_{2}(m)+1$, then viewing $\gamma$ as an element in $\mathbb{Z}/2m\mathbb{Z}$, the $q$-cyclotomic coset modulo $2m$ containing $\gamma$ is 
		$$c_{2m/q}(\gamma) =  \{\gamma,\gamma q,\cdots,\gamma q^{2\tau-1}\}.$$
		Moreover, the coset $c_{2m/q}(\gamma)$ contains $m+\gamma$, and is exactly the preimage of $c_{m/q}(\gamma)$ under the projection $\pi_{2m/m}$.
	\end{itemize}
\end{lemma}

\begin{proof}
	\begin{itemize}
		\item[(1)] Since $\tau$ is the smallest positive integer such that $\gamma q^{\tau} \equiv \gamma \pmod{m}$ and $2m \mid \gamma q^{\tau}-\gamma$, then $\tau$ is also the smallest positive integer such that $\gamma q^{\tau} \equiv \gamma \pmod{2m}$. Therefore we have
		$$c_{2m/q}(\gamma) = \{\gamma,\gamma q,\cdots,\gamma q^{\tau-1}\}.$$
		
		Also by $2m \mid \gamma q^{\tau}-\gamma$, the congruence equation $\gamma q^{x} \equiv m+\gamma \pmod{2m}$ has no solution. It follows that $c_{2m/q}(\gamma)$ and $c_{2m/q}(n+\gamma)$ are disjoint. As $q$ is odd,
		$$(m+\gamma)q^{j} = mq^{j} + \gamma q^{j} \equiv m + \gamma q^{j} \pmod{2m}$$
		for any $j \in \mathbb{N}$, which implies that $(m+\gamma)q^{j} \equiv m+\gamma \pmod{2m}$ if and only if $\gamma q^{j} \equiv \gamma \pmod{2m}$. Thus we obtain
		$$c_{2m/q}(m+\gamma) = \{m+\gamma,(m+\gamma)q,\cdots,(m+\gamma)q^{\tau-1}\} = \{m+\gamma,m+\gamma q,\cdots,m+\gamma q^{\tau-1}\}.$$
		It is straightforward to check that the disjoint union $c_{2m/q}(\gamma) \sqcup c_{2m/q}(n+\gamma)$ is exactly the preimage of $c_{m/q}(\gamma)$ under $\pi_{2m/m}$.
		
		\item[(2)] Since $m \mid \gamma q^{\tau}-\gamma$ and $2m \nmid \gamma q^{\tau}-\gamma$, we have
		$$v_{2}(\gamma q^{\tau}-\gamma) = v_{2}(q^{\tau}-1) + v_{2}(\gamma) = v_{2}(m).$$
		It follows from Lemma \ref{lem 2} that $2\tau$ is the smallest positive integer such that 
		$$\gamma q^{2\tau} \equiv \gamma \pmod{2m}.$$
		Hence the $q$-cyclotomic coset modulo $2m$ containing $\gamma$ is 
		$$c_{2m/q}(\gamma) = \{\gamma,\gamma q,\cdots,\gamma q^{2\tau-1}\}.$$
		
		Also by $m \mid \gamma q^{\tau}-\gamma$ and $2m \nmid \gamma q^{\tau}-\gamma$, we obtain
		$$\gamma q^{\tau} \equiv m+\gamma \pmod{2m},$$
		which implies that $m+\gamma \in c_{2m/q}(\gamma)$. Furthermore, it still holds that $(m+\gamma)q^{j} \equiv m + \gamma q^{j} \pmod{2m}$ for any $j \in \mathbb{N}$. Hence the coset
		$$c_{2m/q}(\gamma) = c_{2m/q}(m+\gamma) = \{\gamma,\gamma q,\cdots,\gamma q^{\tau-1},m+\gamma,m+\gamma q,\cdots,m+\gamma q^{\tau-1}\}$$
		is exactly the preimage of $c_{m/q}(\gamma)$ under $\pi_{2m/m}$.
	\end{itemize}
\end{proof}

If $\gamma \in \mathbb{Z}/m\mathbb{Z}$ fits the first condition in Lemma \ref{lem 4}, that is,
$$\pi_{2m/m}^{-1}(c_{m/q}(\gamma)) = c_{2m/q}(\gamma) \sqcup c_{2m/q}(m+\gamma),$$
then we say that $c_{m/q}(\gamma)$ is splitting with respect to the extension $\mathbb{Z}/2m\mathbb{Z} : \mathbb{Z}/m\mathbb{Z}$. Otherwise, $c_{m/q}(\gamma)$ satisfies
$$\pi_{2m/m}^{-1}(c_{m/q}(\gamma)) = c_{2m/q}(\gamma) = c_{2m/q}(m+\gamma),$$
in which case we say that $c_{m/q}(\gamma)$ is nonsplitting or stable with respect to the extension $\mathbb{Z}/2m\mathbb{Z} : \mathbb{Z}/m\mathbb{Z}$. It can be verified that whether or not a cyclotomic coset split does not depend on the choice of the representative of the coset, thus the above definitions are well-defined.

\begin{corollary}\label{coro 1}
	If $n= p_{1}^{e_{1}}\cdots p_{s}^{e_{s}}$ is an odd integer, where $p_{1},\cdots,p_{s}$ are distinct odd primes different from $p$ and $e_{1},\cdots,e_{s}$ are positive integers, then all the $q$-cyclotomic cosets modulo $2n$ are given by
	$$c_{2n/q}(\eta^{(y,x)}) = \{\eta^{(y,x)}, \eta^{(y,x)}q, \cdots, \eta^{(y,x)}q^{\tau_{y}-1}\},$$
	and
    $$c_{2n/q}(\eta^{(y,x)}) =\{\eta^{(y,x)}+n, \eta^{(y,x)}q+n, \cdots, \eta^{(y,x)}q^{\tau_{y}-1}+n\},$$
	for all $(y,x) \in \Sigma(n;q)$, where $\eta = (\eta_{1},\cdots,\eta_{s})$ is a primitive root system modulo $n$.
\end{corollary}

\begin{proof}
	Since $q$ is odd, for any $(y,x) \in \Sigma(n;q)$
	$$v_{2}(\eta^{(y,x)}) + v_{2}(q^{\tau_{y_{1} \cdots y_{s}}}) \geq 1,$$
	therefore the coset $c_{n/q}(\eta^{(y,x)})$ is splitting with respect to $\mathbb{Z}/2n\mathbb{Z} : \mathbb{Z}/n\mathbb{Z}$.
\end{proof}

Let $m$ be a positive integer such that $v=v_{2}(m) > 0$. For $\gamma \in \mathbb{Z}/m\mathbb{Z}$ it makes sense to say whether $\gamma$ is divisible by $2^{v}$ or not, as it does not depend on the choice of the representative of $\gamma$ in $\mathbb{Z}$. If $\gamma$ is divisible by $2^{v}$, then the coset $c_{m/q}(\gamma)$ splits into $c_{2m/q}(\gamma)$ and $c_{2m/q}(m+\gamma)$ along $\pi_{2m/m}$. It is trivial to check that exact one among $\gamma$ and $m+\gamma$ is divisible by $2^{v+1}$. Applying this argument successively gives rise to the following definition.

\begin{definition}\label{def 1}
	Let $n$ be an odd integer. For any $\gamma \in \mathbb{Z}$, there exists a unique power series
	$$\varphi_{n}(\gamma) = \sum_{k=0}^{\infty}\phi_{n,k}(\gamma)\cdot 2^{k}, \ \phi_{n,k}(\gamma) \in \{0,1\},$$
	such that for every $i \in \mathbb{N}$ it holds that $2^{i+1} \mid \gamma+ n\cdot \varphi_{n}(\gamma)_{\leq i}$, where $\varphi_{n}(\gamma)_{\leq i} = \sum\limits_{k=0}^{i}\phi_{n,k}(\gamma)\cdot 2^{k}$ is the finite sum of the first $i+1$ terms of $\varphi_{n}(\gamma)$. The power series $\varphi_{n}(\gamma)$ is called the cyclotomic $2$-adic integer associated to $\gamma$.
	
	Let $\overline{\gamma} \in \mathbb{Z}/n\mathbb{Z}$, and $\gamma$ be a preimage of $\overline{\gamma}$ in $\mathbb{Z}$. The cyclotomic $2$-adic integer $\varphi_{n}(\gamma)$ is said to be over $\overline{\gamma}$.
\end{definition}

\begin{remark}
	\begin{itemize}
		\item [(1)] It is well-known fact that there is a canonical isomorphism from the ring of $2$-adic power series onto the ring $\mathbb{Z}_{2}$ of $2$-adic integers, which justify the name cyclotomic $2$-adic integer.
		\item [(2)] To define the $2$-adic integer associated to a residue class $\overline{\gamma} \in \mathbb{Z}/n\mathbb{Z}$, without specifying a representative of $\overline{\gamma}$, one should work with the quotient of $\mathbb{Z}_{2}$ modulo the subgroup $\mathbb{Z}$. In this paper we will follow Definition \ref{def 1}, which suffices for our purpose.
	\end{itemize}
\end{remark}

The following figure shows the cyclotomic $2$-adic integer associated to $\gamma$, which is represented by the solid line.

\begin{center}
	\begin{tikzpicture}
		\node{$\gamma$}
		child{node {$\gamma$}
			child{node {$\gamma$} edge from parent[dashed]
				child{node {$\gamma$}
					child{node{$\vdots$}}
					child{node{$\vdots$}
						child[missing]{}}}
				child{node {$4n^{\prime}+\gamma$}
					child{node{$\vdots$}}
					child{node{$\vdots$}}}}
			child[missing]{}
			child{node {$2n^{\prime}+\gamma$}
				child{node {$2n^{\prime}+\gamma$}
					child{node{$\vdots$}edge from parent[dashed]}
					child{node{$\vdots$}}}
				child{node {$6n^{\prime}+\gamma$} edge from parent[dashed]
					child{node{$\vdots$}}
					child{node{$\vdots$}}}}}
		child[missing]{}
		child[missing]{}
		child[missing]{}
		child{node {$n^{\prime}+\gamma$} edge from parent[dashed]
			child{node {$n^{\prime}+\gamma$}
				child{node {$n^{\prime}+\gamma$}
					child{node{$\vdots$}}
					child{node{$\vdots$}
						child[missing]{}}}
				child{node {$5n^{\prime}+\gamma$}
					child{node{$\vdots$}}
					child{node{$\vdots$}}}}
			child[missing]{}
			child{node {$3n^{\prime}+\gamma$}
				child{node {$3n^{\prime}+\gamma$}
					child{node{$\vdots$}}
					child{node{$\vdots$}}}
				child{node {$7n^{\prime}+\gamma$}
					child{node{$\vdots$}}
					child{node{$\vdots$}}}}}
		;
	\end{tikzpicture}
\end{center}
$$\mathrm{Fig.} \ 1$$

Let $c_{n/q}(\gamma)$ be a $q$-cyclotomic coset modulo $n$, and 
$$\varphi_{n}(\gamma) = \sum_{k=0}^{\infty}\phi_{n,k}(\gamma)\cdot 2^{k}$$
be the cyclotomic $2$-adic integer associated to $\gamma$. For every nonnegative integer $d$ we define an integer $U_{d}(\gamma) = \sum\limits_{k=0}^{d}u_{d,k}(\gamma)\cdot 2^{k}$, $u_{d,k}(\gamma) \in \{0,1\}$, by
\begin{equation*}
	\left\{
	\begin{array}{lcl}
		u_{d,k}(\gamma) = \phi_{n,k}(\gamma), \ 0 \leq k \leq d-1;\\
		u_{d,k}(\gamma) = 1-\phi_{n,k}(\gamma), \ k = d.
	\end{array} \right.
\end{equation*}
We call $U_{d}(\gamma)$ the $d$-th generator of the cyclotomic system $\mathcal{PC}_{n/q}(\gamma)$, and call $\varphi_{n}(\gamma)$ the generator at infinity of $\mathcal{PC}_{n/q}(\gamma)$.

The following theorems determine precisely the elements in $\mathcal{PC}_{n/q}(\gamma)$, in the case where $q^{\tau} \equiv 1 \pmod{4}$ and where $q^{\tau} \equiv 3 \pmod{4}$ respectively. Let $(c_{2^{i}n/q}(\gamma_{i}))_{i \in \mathbb{N}}$ be an element in $\mathcal{PC}_{n/q}(\gamma)$. If the component $c_{2^{i_{0}}n/q}(\gamma_{i_{0}})$ splits with respect to the extension $\mathbb{Z}/2^{i_{0}+1}n\mathbb{Z}: \mathbb{Z}/2^{i_{0}}n\mathbb{Z}$, we say that $(c_{2^{i}n/q}(\gamma_{i}))_{i \in \mathbb{N}}$ splits at degree $i_{0}$. If $(c_{2^{i}n/q}(\gamma_{i}))_{i \in \mathbb{N}}$ splits at every degree, it is said to be splitting. Otherwise, it is called stable.

\begin{theorem}\label{thm 5}
	Let $c_{n/q}(\gamma)$ be a $q$-cyclotomic coset modulo $n$ with $| c_{n/q}(\gamma) | =\tau$. Assume that $q^{\tau} \equiv 1 \pmod{4}$. Then
	\begin{description}
		\item[(1)] The generator $\varphi_{n}(\gamma)$ at infinity gives rise to an element
		$$c_{\infty} = (c_{2^{i}n/q}(\gamma + n\cdot \varphi_{n}(\gamma)_{\leq i-1}))_{i \in \mathbb{N}} \in \mathcal{PC}_{n/q}(\gamma),$$
		where $\varphi_{n}(\gamma)_{\leq -1}$ is set to be $0$. It is the unique splitting element in $\mathcal{PC}_{n/q}(\gamma)$.
		\item[(2)] For each $d \in \mathbb{N}$, the $d$-th generator $U_{d}(\gamma)$ gives rise to $2^{v^{+}-1}$ elements
		$$c_{d,t} = (c_{2^{i}n/q}(\gamma + n\cdot \widehat{U}_{d,t}(\gamma)_{\leq i-1}))_{i \in \mathbb{N}} \in \mathcal{PC}_{n/q}(\gamma),$$
		where $v^{+} = v_{2}(q^{\tau}-1)$, $t = (t_{1},\cdots,t_{v^{+}-1}) \in \{0,1\}^{v^{+}-1}$, and 
		$$\widehat{U}_{d,t}(\gamma) = U_{d}(\gamma) + \sum_{j=1}^{v^{+}-1}t_{j}\cdot 2^{d+j}.$$
		In particular, we set $\widehat{U}_{d,t}(\gamma)_{\leq-1}$ to be $0$.
	\end{description}
	Furthermore, all elements in $\mathcal{PC}_{n/q}(\gamma)$ are obtained by the above construction unrepeatedly.
\end{theorem}

\begin{proof}
	It is trivial to see that $c_{\infty}$ and all $c_{d,t}$'s are compatible sequences of $q$-cyclotomic cosets lying in $\mathcal{PC}_{n/q}(\gamma)$. First we prove that they are pairwise distinct. For each $i \geq 1$, by the definition of $\varphi_{n}(\gamma)$ we have
	$$v_{2}(\gamma+n\cdot \varphi_{n}(\gamma)_{\leq i-1})+v_{2}(q^{\tau}-1) > i+1,$$
	which indicates that $c_{2^{i}n/q}(\gamma+n\cdot \varphi_{n}(\gamma)_{\leq i-1})$ splits with respect to $\mathbb{Z}/2^{i+1}n\mathbb{Z}: \mathbb{Z}/2^{i}n\mathbb{Z}$. Hence $c_{\infty}$ is a splitting sequences.
	
	For any $d \in \mathbb{N}$ and any $t = (t_{1},\cdots,t_{v^{+}-1}) \in \{0,1\}^{v^{+}-1}$, by the definition of $U_{d}(\gamma)$ it holds that 
	\begin{equation*}
		v_{2}(\gamma+n\cdot \widehat{U}_{d,t}(\gamma)_{\leq i-1}) + v_{2}(q^{\tau}-1) \left\{
		\begin{array}{lcl}
			\geq i+v^{+}, \ 1\leq i \leq d;\\
			=d+v^{+}, \ i \geq d+1.
		\end{array} \right.
	\end{equation*}
	Consequently, $c_{d,t}$ splits at degree less than $d+v^{+}$, while begin to be nonsplitting from degree $d+v^{+}$. It follows from that $c_{\infty}$ and $c_{d,t}$'s are pairwise different.
	
    Next we prove that the sequences $c_{\infty}$ and $c_{d,t}$'s are exactly all the elements in $\mathcal{PC}_{n/q}(\gamma)$. Let $(c_{2^{i}n/q}(\gamma_{i}))_{i \in \mathbb{N}} \in \mathcal{PC}_{n/q}(\gamma)$. By Lemma \ref{lem 4} the representative $\gamma_{i}$ can be chosen to be in the form
	$$\gamma_{i} = \gamma+n\cdot Y_{\leq i-1},$$
	where $Y_{\leq i-1}$ is the finite sum of the first $i$ terms of a power series $Y = \sum\limits_{k=0}^{\infty} y_{k}\cdot 2^{k}$. If $Y = \varphi_{n}(\gamma)$ then $(c_{2^{i}n/q}(\gamma_{i}))_{i \in \mathbb{N}}$ is just $c_{\infty}$. Otherwise there is a smallest nonnegative integer $d$ such that $y_{d} \neq \phi_{n,d}(\gamma)$, which implies that 
	$$\gamma_{d+1} = \gamma + n\cdot Y_{\leq d} = \gamma + n \cdot U_{d}(\gamma).$$
	Applying Lemma \ref{lem 4} indicates that all the elements in $\mathcal{PC}_{n/q}(\gamma)$ with the component at degree $d+1$ being $c_{2^{d+1}n/q}(\gamma+n\cdot U_{d}(\gamma))$ are exactly $c_{d,t}$ for all $t = (t_{1},\cdots,t_{v^{+}-1}) \in \{0,1\}^{v^{+}-1}$. Here we complete the proof.
\end{proof}

\begin{theorem}\label{thm 6}
	Let $c_{n/q}(\gamma)$ be a $q$-cyclotomic coset modulo $n$ with $| c_{n/q}(\gamma) | =\tau$. Assume that $q^{\tau} \equiv 3 \pmod{4}$. Then 
	\begin{description}
		\item[(1)] The generator $\varphi_{n}(\gamma)$ at infinity gives rise to an element
		$$c_{\infty} = (c_{2^{i}n/q}(\gamma + n\cdot \varphi_{n}(\gamma)_{\leq i-1}))_{i \in \mathbb{N}} \in \mathcal{PC}_{n/q}(\gamma),$$
		where $\varphi_{n}(\gamma)_{\leq -1}$ is set to be $0$. It is the unique splitting element in $\mathcal{PC}_{n/q}(\gamma)$.
		\item[(2)] For each $d \in \mathbb{N}$, the $d$-th generator $U_{d}(\gamma)$ gives rise to $2^{v^{-}-1}$ elements
		$$c_{d,t} = (c_{2^{i}n/q}(\gamma + n\cdot \widehat{U}_{d,t}(\gamma)_{\leq i-1}))_{i \in \mathbb{N}} \in \mathcal{PC}_{n/q}(\gamma),$$
		where $v^{-} = v_{2}(q^{\tau}+1)$, $t=(t_{1},\cdots,t_{v^{-}-1}) \in \{0,1\}^{v^{-}-1}$, and 
		$$\widehat{U}_{d,t}(\gamma) = U_{d}(\gamma) + \sum_{j=1}^{v^{-}-1}t_{j}\cdot 2^{d+j+1}.$$
		In particular, we set $\widehat{U}_{d,t}(\gamma)_{\leq -1}$ to be $0$.
	\end{description}
	Furthermore, all elements in $\mathcal{PC}_{n/q}(\gamma)$ are obtained by the above construction unrepeatedly.
\end{theorem}

\begin{proof}
	Let $d \in \mathbb{N}$ and $t = (t_{1},\cdots,t_{v^{-}-1}) \in \{0,1\}^{v^{-}-1}$. Denote by $\tau_{i} =|c_{2^{i}n/q}(\gamma+n\cdot \widehat{U}_{d,t}(\gamma)_{\leq i-1})|$ for $i \in \mathbb{N}$. Notice that 
		\begin{equation*}
		v_{2}(\gamma+n\cdot \widehat{U}_{d,t}(\gamma)_{\leq i-1}) \left\{
		\begin{array}{lcl}
			\geq i, \ 1\leq i \leq d;\\
			=d, \ i \geq d+1.
		\end{array} \right.
	\end{equation*}
	Therefore one has that 
	\begin{itemize}
		\item[(1)] when $0 \leq i \leq d$, $c_{2^{i}n/q}(\gamma + n\cdot \widehat{U}_{d,t}(\gamma)_{\leq i-1})$ splits and $\tau_{i+1} = \tau$;
		\item[(2)] when $i =d+1$, $c_{2^{i}n/q}(\gamma + n\cdot \widehat{U}_{d,t}(\gamma)_{\leq i-1})$ does not split and $\tau_{i+1} = 2\tau$;
		\item[(4)] when $d+2 \leq i \leq d+v^{-}$, $c_{2^{i}n/q}(\gamma + n\cdot \widehat{U}_{d,t}(\gamma)_{\leq i-1})$ splits and $\tau_{i+1} = 2\tau$;
		\item[(5)] when $i \geq d+v^{-}+1$, $c_{2^{i}n/q}(\gamma + n\cdot \widehat{U}_{d,t}(\gamma)_{\leq i-1})$ does not split and $\tau_{i+1} = 2^{i+1-d-v^{-}}\tau$;
	\end{itemize}
	Now Theorem \ref{thm 6} can be proved via the same argument as the proof of Theorem \ref{thm 5}.
\end{proof}

In general, there is a unique element $c_{\infty}$ in the cyclotomic system $\mathcal{PC}_{n/q}(\gamma)$ which splits at every degree. It is called the principal sequence over $c_{n/q}(\gamma)$, while the other elements are called stable sequences. If $c$ is a stable sequence, it is induced by a unique generator $U_{d}(\gamma)$, $d \in \mathbb{N}$. Define the quasi-stable degree of $c$ to be $d+1$, and the stable degree of $c$ to be the degree from which $c$ is always nonsplitting. By the proofs of Theorem \ref{thm 5} and \ref{thm 6}, if $q^{\tau} \equiv 1 \pmod{4}$, the stable degree of $c$ is $d+v^{+}$; if $q^{\tau} \equiv 3 \pmod{4}$, the stable degree of $c$ is $d+v^{-}+1$. 

Figure $2$ and $3$ exhibit the elements in $\mathcal{PC}_{n/q}(\gamma)$, in the case where $q^{\tau} \equiv 1 \pmod{4}$ and where $q^{\tau} \equiv 3 \pmod{4}$ respectively. The red line represents the principal sequence.

\begin{center}
	\begin{tikzpicture}
		\node {$c_{n^{\prime}/q}(\gamma)$}
		child{node {$\bullet$}
		child{node {$\stackrel{\bullet}{\vdots}$}
		child {node {$\bullet$}
		child {node {$\stackrel{\bullet}{\vdots}$}}}
		child {node {$\bullet$}
		child{node {$\stackrel{\bullet}{\vdots}$}}}}
		child{node {$\stackrel{\bullet}{\vdots}$}
		child{node {$\bullet$}
		child{node {$\stackrel{\bullet}{\vdots}$}}}
		child{node {$\bullet$}
		child{node {$\stackrel{\bullet}{\vdots}$}}}}
		}
		child[missing]{}
		child[missing]{}
		child[grow=down, red]{node {$\bullet$} 
		child[grow'=down, red]{node {$\bullet$}
		child[grow'=down, red]{node {$\bullet$}
		child{node {$\bullet$}
	    child{node {$\stackrel{\bullet}{\vdots}$}}}}}
		child[missing]{}
		child[missing]{}
		child[black]{node {$\bullet$}
		child{node {$\stackrel{\bullet}{\vdots}$}
		child{node {$\bullet$}
		child{node {$\stackrel{\bullet}{\vdots}$}}}
		child{node {$\bullet$}
		child{node {$\stackrel{\bullet}{\vdots}$}}}}
		child{node {$\stackrel{\bullet}{\vdots}$}
		child{node {$\bullet$}
		child{node {$\stackrel{\bullet}{\vdots}$}}}
	    child{node {$\bullet$}
        child{node {$\stackrel{\bullet}{\vdots}$}}}}}
		}
		;
	\end{tikzpicture}
\end{center}
$$\mathrm{Fig.} \ 2$$

\begin{center}
	\begin{tikzpicture}
		\node {$c_{n^{\prime}/q}(\gamma)$}
		child{node {$\bullet$}
		child{node {$\bullet$}
		child{node {$\stackrel{\bullet}{\vdots}$}
		child{node {$\bullet$}
		child{node {$\bullet$}
		child{node {$\stackrel{\bullet}{\vdots}$}}}}
		child{node {$\bullet$}
		child{node {$\bullet$}
	    child{node {$\stackrel{\bullet}{\vdots}$}}}}}
		child{node {$\stackrel{\bullet}{\vdots}$}
		child{node {$\bullet$}
		child{node {$\bullet$}
		child{node {$\stackrel{\bullet}{\vdots}$}}}}
	    child{node {$\bullet$}
        child{node {$\bullet$}
        child{node {$\stackrel{\bullet}{\vdots}$}}}}}}
		}
		child[missing]{}
		child[missing]{}
		child[grow=down,red]{node {$\bullet$}
		child[grow'=down]{node {$\bullet$}
		child[grow'=down]{node {$\bullet$}
		child[grow'=down]{node {$\bullet$}
	    child[grow'=down]{node {$\bullet$}
        child[grow'=down]{node {$\stackrel{\bullet}{\vdots}$}}}}}}
		child[missing]{}
		child[missing]{}
		child[black]{node {$\bullet$}
		child{node {$\bullet$}
		child{node {$\stackrel{\bullet}{\vdots}$}
		child{node {$\bullet$}
		child{node {$\stackrel{\bullet}{\vdots}$}}}
		child{node {$\bullet$}
		child{node {$\stackrel{\bullet}{\vdots}$}}}}
	    child{node {$\stackrel{\bullet}{\vdots}$}
        child{node {$\bullet$}
        child{node {$\stackrel{\bullet}{\vdots}$}}}
        child{node {$\bullet$}
        child{node {$\stackrel{\bullet}{\vdots}$}}}}}}
		}
		;
	\end{tikzpicture}
\end{center}
$$\mathrm{Fig.} \ 3$$

Combining Proposition \ref{prop 4}, Theorem \ref{thm 5} and \ref{thm 6}, we give the size of each component of the sequences in $\mathcal{PC}_{n/q}(\gamma)$. Fix a primitive root system $\eta = (\eta_{1},\cdots,\eta_{s})$ mdoulo $n$. The representative $\gamma$ can be chosen to be $\gamma = \eta^{(y,x)}$ for some $(y,x) \in \Sigma(n;q)$. 

\begin{corollary}\label{coro 4}
	\begin{itemize}
		\item[(1)] If $(y,x) \in \Sigma(n;q)$ satisfies that $q^{\tau_{y}} \equiv 1 \pmod{4}$, then
		\begin{itemize}
			\item[(1a)] The principal sequence $c_{\infty}$ has the same size at every degree, given by
			$$| c_{2^{i}n/q}(\gamma + n\cdot \varphi_{n}(\gamma)_{\leq i-1})| =  \omega_{y}\cdot p_{1}^{M_{y_{1}}}\cdots p_{s}^{M_{y_{s}}}.$$
			\item[(1b)] For any $d \in \mathbb{N}$ and any $t = (t_{1},\cdots,t_{v^{+}-1}) \in \{0,1\}^{v^{+}-1}$, the size of each component of $c_{d,t}$ is given by
			\begin{small}
				\begin{equation*}
					|c_{2^{i}n/q}(\gamma + n\cdot \widehat{U}_{d,t}(\gamma)_{\leq i-1}) | = \left\{
					\begin{array}{lcl}
						\omega_{y}\cdot p_{1}^{M_{y_{1}}}\cdots p_{s}^{M_{y_{s}}}, \quad 0 \leq i \leq d+v^{+};\\
						2^{i-d-v^{+}}\omega_{y}\cdot p_{1}^{M_{y_{1}}}\cdots p_{s}^{M_{y_{s}}}, \quad i \geq d+v^{+}+1.
					\end{array} \right.
				\end{equation*}
			\end{small}
		\end{itemize}
		\item[(2)] If $(y,x) \in \Sigma(n;q)$ satisfies that $q^{\tau_{y}} \equiv 3 \pmod{4}$, then
		\begin{itemize}
			\item[(2a)] The principal sequence $c_{\infty}$ has the same size at every degree, given by
			$$| c_{2^{i}n/q}(\gamma + n\cdot \varphi_{n}(\gamma)_{\leq i-1})| =  \omega_{y}\cdot p_{1}^{M_{y_{1}}}\cdots p_{s}^{M_{y_{s}}}.$$
			\item[(2b)] For any $d \in \mathbb{N}$ and any $t = (t_{1},\cdots,t_{v^{+}-1}) \in \{0,1\}^{v^{-}-1}$, the size of each component of $c_{d,t}$ is given by
			\begin{small}
				\begin{equation*}
					| c_{2^{i}n/q}(\gamma + n\cdot \widehat{U}_{d,t}(\gamma)_{\leq i-1}) | = \left\{
					\begin{array}{lcl}
						\omega_{y}\cdot p_{1}^{M_{y_{1}}}\cdots p_{s}^{M_{y_{s}}}, \quad 0 \leq i \leq d+1;\\
						2\omega_{y}\cdot p_{1}^{M_{y_{1}}}\cdots p_{s}^{M_{y_{s}}}, \quad d+2 \leq i \leq d+v^{-}+1;\\
						2^{i-d-v^{-}}\omega_{y}\cdot p_{1}^{M_{y_{1}}}\cdots p_{s}^{M_{y_{s}}}, \quad i \geq d+v^{-}+2.
					\end{array} \right.
				\end{equation*}
			\end{small}
		\end{itemize}
	\end{itemize}
\end{corollary}

\begin{proof}
	By Proposition \ref{prop 4} the size of the coset $c_{n/q}(\eta^{(y,x)})$ is 
	$$\tau_{y} = \omega_{y}\cdot p_{1}^{M_{y_{1}}}\cdots p_{s}^{M_{y_{s}}}.$$
	Notice that if a sequence $(c_{2^{i}n/q}(\gamma_{i}))_{i \in \mathbb{N}}$ splits at degree $i$ then $| c_{2^{i+1}n/q}(\gamma_{i+1}) | = | c_{2^{i}n/q}(\gamma_{i}) |$; otherwise $| c_{2^{i+1}n/q}(\gamma_{i+1}) | = 2| c_{2^{i}n/q}(\gamma_{i}) |$. Thus the assertions follow from the proofs of Theorem \ref{thm 5} and \ref{thm 6}.
\end{proof}

\subsubsection{The representatives and the sizes of $q$-cyclotomic cosets modulo $n$}\label{sec 8}
With the results obtained in the last two subsections we give a full set of representatives of $\mathcal{C}_{n/q}$ and the sizes of the cosets in $\mathcal{C}_{n/q}$. In this subsection it is assumed that $q$ is a power of an odd prime $p$, $n = 2^{e_{0}}p_{1}^{e_{1}}\cdots p_{s}^{e_{s}}$ is an even integer, where $p_{1},\cdots,p_{s}$ are distinct odd primes different from $p$ and $e_{0},e_{1},\cdots,e_{s}$ are positive integers, and $n^{\prime} = p_{1}^{e_{1}}\cdots p_{s}^{e_{s}}$ is the maximal odd divisor of $n$. 

Let $\eta = (\eta_{1},\cdots,\eta_{s})$ be a system of primitive roots modulo $n^{\prime}$, and $\Sigma(n^{\prime};q)$ be the set of applicable tuples with respect to $n^{\prime}$ and $q$. Further, we set
$$\Sigma^{+}(n^{\prime};q) = \{(y,x) \in \Sigma(n^{\prime};q) \ | \ q^{\tau_{y}} \equiv 1 \pmod{4}\}$$
and
$$\Sigma^{-}(n^{\prime};q) = \{(y,x) \in \Sigma(n^{\prime};q) \ | \ q^{\tau_{y}} \equiv 3 \pmod{4}\}.$$
Recall that for any $(y,x) = (y_{1},\cdots,y_{s},x_{1},\cdots,x_{s}) \in \Sigma$,
$$\tau_{y} = \omega_{y}\cdot p_{1}^{M_{y_{1}}}\cdots p_{s}^{M_{y_{s}}}$$
is the size of $c_{n^{\prime}/q}(\eta^{(y,x)})$, where $m_{y_{i}} = \mathrm{min}\{e_{i}-y_{i},1\}$, $\omega_{y}=\mathrm{lcm}(\mathrm{ord}_{p_{1}^{m_{y_{1}}}}(q),\cdots,\mathrm{ord}_{p_{s}^{m_{y_{s}}}}(q))$ and $M_{y_{i}} = \mathrm{max}\{e_{i}-y_{i}-v_{p_{i}}(q^{\omega_{y}}-1),0\}$, $i=1,\cdots,s$. If $(y,x) \in \Sigma^{+}(n^{\prime};q)$, then we set $v^{+}(y) = v_{2}(q^{\tau_{y}}-1)$ and for $0 \leq d \leq e_{0}$
$$s_{y,d}^{+}=\mathrm{min}\{e_{0}-d-1,v^{+}(y)-1\};$$
otherwise, if $(y,x) \in \Sigma^{-}(n^{\prime};q)$, then we set $v^{-}(y) = v_{2}(q^{\tau_{y}}+1)$ and for $0 \leq d \leq e_{0}$
$$s_{y,d}^{-}=\mathrm{min}\{e_{0}-d-2,v^{-}(y)-1\}.$$

\begin{definition}
	\begin{description}
		\item[(1)] Let $(y,x) \in \Sigma^{+}(n^{\prime};q)$. For any $d \in \{0,1,\cdots,e_{0}\}$ and $t = (t_{1},\cdots,t_{s_{y,d}^{+}}) \in \{0,1\}^{s_{y,d}^{+}}$ (if $s_{y,d}^{+} \leq 0$ then $\{0,1\}^{s_{y,d}^{+}}$ is set to be $\{0\}$), define
		$$\mu(y,x)_{d,t}^{+} = \eta^{(y,x)} + n^{\prime}\cdot(U_{d}(\eta^{(y,x)})+ \sum_{j=1}^{s_{y,d}^{+}}t_{j}\cdot 2^{d+j}),$$
		and define $\Sigma_{(y,x)}$ to be the set of all these $\mu(y,x)_{d,t}^{+}$:
		$$\mathcal{CR}_{(y,x)}^{+} = \{\mu(y,x)_{d,t}^{+} \ | \ 0 \leq d \leq e_{0}, t \in \{0,1\}^{s_{y,d}^{+}}\}.$$
		\item[(2)] Let $(y,x) \in \Sigma^{-}(n^{\prime};q)$. For any $d \in \{0,1,\cdots,e_{0}\}$ and $t = (t_{1},\cdots,t_{s_{y,d}^{-}}) \in \{0,1\}^{s_{y,d}^{-}}$ (if $s_{y,d}^{-} \leq 0$ then $\{0,1\}^{s_{y,d}^{-}}$ is set to be $\{0\}$), define
		$$\mu(y,x)_{d,t}^{-} = \eta^{(y,x)} + n^{\prime}\cdot(U_{d}(\eta^{(y,x)})+ \sum_{j=1}^{s_{y,d}^{-}}t_{j}\cdot 2^{d+j+1}),$$
		and define $\Sigma_{(y,x)}$ to be the set of all these $\mu(y,x)_{d,t}^{-}$:
		$$\mathcal{CR}_{(y,x)}^{-} = \{\mu(y,x)_{d,t}^{-} \ | \ 0 \leq d \leq e_{0}, t \in \{0,1\}^{s_{y,d}^{-}}\}.$$
	\end{description}
\end{definition}

\begin{theorem}\label{thm 7}
	A full set of representatives of $q$-cyclotomic cosets modulo $n$ is given by
	$$\mathcal{CR}_{n/q} = (\bigsqcup_{(y,x) \in \Sigma^{+}(n^{\prime};q)} \mathcal{CR}_{(y,x)}^{+}) \sqcup (\bigsqcup_{(y,x) \in \Sigma^{-}(n^{\prime};q)} \mathcal{CR}_{(y,x)}^{-}).$$
\end{theorem}

\begin{proof}
	By Proposition \ref{prop 1} we have the decomposition
	$$\mathcal{C}_{n/q} = \bigsqcup_{(y,x) \in \Sigma(n^{\prime};q)}\widehat{\pi}_{n/n^{\prime}}^{-1}(c_{n^{\prime}/q}(\eta^{(y,x)})).$$
	For any $(y,x) \in \Sigma(n^{\prime};q)$, the $q$-cyclotomic cosets modulo $n$ lying in $\widehat{\pi}_{n/n^{\prime}}^{-1}(c_{n^{\prime}/q}(\eta^{(y,x)}))$ can be obtained via taking the components at degree $e_{0}$ of the sequences with quasi-stable degree $\leq e_{0}+1$ in $\mathcal{PC}_{n^{\prime}/q}(\eta^{(y,x)})$, since all the sequences with quasi-stable degree $\geq e_{0}+1$ have the same component at degree $e_{0}$. To be explicit, if $(y,x) \in \Sigma^{+}(y,x)$, then the representatives of the cosets in $\widehat{\pi}_{n/n^{\prime}}^{-1}(c_{n^{\prime}/q}(\eta^{(y,x)}))$ are given by
	$$\eta^{(y,x)} + n\cdot\widehat{U}_{d,t}(\eta^{(y,x)})_{\leq e_{0}-1} = \eta^{(y,x)} + n\cdot(U_{d}(\eta^{(y,x)})+ \sum_{j=1}^{s_{y,d}^{+}}t_{j}\cdot 2^{d+j}) = \mu(y,x)_{d,t}^{+},$$
	for all $0 \leq d \leq e_{0}$ and $t = (t_{1},\cdots,t_{s_{y,d}^{+}}) \in \{0,1\}^{s_{y,d}^{+}}$. Similar argument applies for the applicable tuples $(y,x) \in \Sigma^{-}(y,x)$.
\end{proof}

\begin{corollary}\label{coro 3}
	\begin{description}
		\item[(1)] Let $(y,x) \in \Sigma^{+}(n^{\prime};q)$. For any $d \in \{0,1,\cdots,e_{0}\}$ and $t = (t_{1},\cdots,t_{s_{y,d}^{+}}) \in \{0,1\}^{s_{y,d}^{+}}$, the size of the coset $c_{n/q}(\mu(y,x)_{d,t}^{+})$ is 
		$$| c_{n/q}(\mu(y,x)_{d,t}^{+}) | = 2^{\mathrm{max}\{e_{0}-d-v^{+}(y),0\}}\cdot\omega_{y}p_{1}^{M_{y_{1}}}\cdots p_{s}^{M_{y_{s}}}.$$
		\item[(2)] Let $(y,x) \in \Sigma^{-}(n^{\prime};q)$. For any $d \in \{0,1,\cdots,e_{0}\}$ and $t = (t_{1},\cdots,t_{s_{y,d}^{-}}) \in \{0,1\}^{s_{y,d}^{-}}$, the size of the coset $c_{n/q}(\mu(y,x)_{d,t}^{-})$ is 
		\begin{equation*}
			| c_{n/q}(\mu(y,x)_{d,t}^{-}) | = \left\{
			\begin{array}{lcl}
				2^{\mathrm{max}\{e_{0}-d-v^{-}(y),1\}}\cdot\omega_{y}p_{1}^{M_{y_{1}}}\cdots p_{s}^{M_{y_{s}}}, \ \mathrm{if} \ 0 \leq d \leq e_{0}-2;\\
				\omega_{y}p_{1}^{M_{y_{1}}}\cdots p_{s}^{M_{y_{s}}}, \ \mathrm{if} \ d=e_{0}-1 \ \mathrm{or} \ e_{0}.
			\end{array} \right.
		\end{equation*}
	\end{description}
\end{corollary}

\begin{proof}
	The conclusions follow from Corollary \ref{coro 4} and Theorem \ref{thm 7}.
\end{proof}

In practice, it is sometimes more convenient to compute the $q$-cyclotomic cosets modulo $n$ inductively from the cosets modulo the maximal odd divisor $n^{\prime}$ of $n$. We exhibit a concrete examples.

\begin{examples}
	Let $q = 3$ and $n = 40$. We begin with the $3$-cyclotomic cosets modulo $5$, which are
	$$\{0\}, \ \{1,3,4,2\}.$$
	It is easy to see that all these cosets split with respect to the extension $\mathbb{Z}/10\mathbb{Z}: \mathbb{Z}/5\mathbb{Z}$. Thus the $3$-cyclotomic cosets modulo $10$ are
	$$\{0\}, \ \{5\}, \ \{1,3,9,7\}, \ \{6,8,4,2\}.$$
	The coset $\{0\}$, $\{1,3,9,7\}$ and $\{6,8,4,2\}$ split with respect to $\mathbb{Z}/20\mathbb{Z}: \mathbb{Z}/105\mathbb{Z}$, while $\{5\}$ does not, which determines the $3$-cyclotomic cosets modulo $20$ as
	$$\{0\}, \ \{10\}, \ \{5,15\}, \ \{1,3,9,7\}, \ \{11,13,19,17\}, \ \{6,18,14,2\}, \ \{16,8,4,12\}.$$
	Finally, as $\{0\}$, $\{5,15\}$, $\{1,3,9,7\}$, $\{11,13,19,17\}$, $\{6,18,14,2\}$ and $\{16,8,4,12\}$ are splitting with respect to $\mathbb{Z}/40\mathbb{Z}: \mathbb{Z}/205\mathbb{Z}$, while $\{10\}$ is nonsplitting, the $3$-cyclotomic cosets modulo $40$ are
	\begin{align*}
		&\{0\}, \ \{20\}, \{10,30\}, \ \{5,15\}, \ \{25,35\}, \ \{1,3,9,27\}, \ \{21,23,29,7\}, \\ 
		&\{11,33,19,17\}, \ \{31,13,39,37\}, \ \{6,18,14,2\}, \ \{26,38,34,22\}, \ \{16,8,24,32\}, \ \{36,28,4,12\}.
	\end{align*}
\end{examples}

\section{Irreducible factorizations of $X^{n}-1$ and of cyclotomic polynomials over $\mathbb{F}_{q}$}
In \cite{Graner}, A. Graner gives general formulas for the irreducible factorizations of $X^{n}-1$ and of cyclotomic polynomials over $\mathbb{F}_{q}$. Speaking roughly, the idea is to first consider certain finite extension of $\mathbb{F}_{q}$ over which $X^{n}-1$ is factorized into irreducible binomials, according to Wu and Yue's theorems, then to compute the minimal polynomials over $\mathbb{F}_{q}$ of these binomials respectively. Based on the results in the last section and the multiple equal-difference structure of cyclotomic cosets, we give improved formula to factorize $X^{n}-1$ and $\Phi_{n}(X)$ over $\mathbb{F}_{q}$ in this section. The improvements are made in the following two aspects. With the full set of representatives and the sizes of $q$-cyclotomic cosets modulo $n$, the irreducible factors of $X^{n}-1$ and of $\Phi_{n}(X)$ can be determined precisely. And the coarsest multiple equal-difference representation of each cyclotomic coset leads to the extension fields of $\mathbb{F}_{q}$ with the lowest degrees over which the irreducible factors of $X^{n}-1$ over $\mathbb{F}_{q}$ can be further factorized into irreducible binomials, which makes these formulas more convenient to apply.

Let $c_{n/q}(\gamma)$ be a $q$-cyclotomic coset modulo $n$ with $|c_{n/q}(\gamma)| = \tau_{\gamma}$. Let $n_{\gamma} = \frac{n}{\mathrm{gcd}(\gamma,n)}$ and 
\begin{equation*}
	\omega_{\gamma} = \left\{
	\begin{array}{lcl}
		2\mathrm{ord}_{\mathrm{rad}(n_{\gamma})}(q), \quad \mathrm{if} \ q^{\mathrm{ord}_{\mathrm{rad}(n_{\gamma})}(q)} \equiv 3 \pmod{4} \ \mathrm{and} \ 8 \mid n_{\gamma};\\
		\mathrm{ord}_{\mathrm{rad}(n_{\gamma})}(q), \quad \mathrm{otherwise}.
	\end{array} \right.
\end{equation*}
By Corollary \ref{coro 2}
$$c_{n/q}(\gamma) = \bigsqcup_{j=0}^{\omega_{\gamma}-1}c_{n/q^{\omega_{\gamma}}}(\gamma q^{j})$$
is the coarsest multiple equal-difference representation of $c_{n/q}(\gamma)$. Notice that for all $0 \leq j \leq \omega_{\gamma}-1$ the cosets $c_{n/q^{\omega_{\gamma}}}(\gamma q^{j})$ have the same size $\frac{\tau_{\gamma}}{\omega_{\gamma}}$, thus they are equal-difference cosets with common difference $d = \frac{n\omega_{\gamma}}{\tau_{\gamma}}$. It follows that the induced irreducible polynomial $M_{\gamma q^{j},q^{\omega_{\gamma}}}(X)$ over $\mathbb{F}_{q^{\omega_{\gamma}}}$ is given by
$$M_{\gamma q^{j},q^{\omega_{\gamma}}}(X) = (X-\zeta_{n}^{\gamma q^{j}})(X-\zeta_{n}^{\gamma q^{j}+d})\cdots (X-\zeta_{n}^{\gamma q^{j}+(\frac{n}{d}-1)d}) = X^{\frac{n}{d}} - \zeta_{d}^{\gamma q^{j}} = X^{\frac{\tau_{\gamma}}{\omega_{\gamma}}} - \zeta_{d_{\gamma}}^{\widetilde{\gamma}q^{j}},$$
for $0 \leq j \leq \omega_{\gamma}-1$, where $d_{\gamma} = \dfrac{\frac{n\omega_{\gamma}}{\tau_{\gamma}}}{\mathrm{gcd}(\frac{n\omega_{\gamma}}{\tau_{\gamma}},\gamma)}$ and $\widetilde{\gamma} = \dfrac{\gamma}{\mathrm{gcd}(\frac{n\omega_{\gamma}}{\tau_{\gamma}},\gamma)}$. Hence we have 
$$M_{\gamma,q}(X) = \prod_{j=0}^{\omega_{\gamma}-1}M_{\gamma q^{j},q^{\omega_{\gamma}}}(X) = \sum_{i=0}^{\omega_{\gamma}}(-1)^{\omega_{\gamma}-i}\left(\sum_{\substack{U \subseteq \{0,\cdots,\omega_{\gamma}-1\}\\ |U|=\omega_{\gamma}-i}}\prod_{u \in U}\zeta_{d_{\gamma}}^{\widetilde{\gamma}q^{u}}\right)X^{\frac{\tau_{\gamma}}{\omega_{\gamma}}\cdot i}.$$
This proves the following lemma.

\begin{lemma}\label{lem 6}
	Let the notations be defined as above. The irreducible factor of $X^{n}-1$ induced by $c_{n/q}(\gamma)$ is given by
	$$M_{\gamma,q}(X) = \sum_{i=0}^{\omega_{\gamma}}(-1)^{\omega_{\gamma}-i}\left(\sum_{\substack{U \subseteq \{0,\cdots,\omega_{\gamma}-1\}\\ |U|=\omega_{\gamma}-i}}\prod_{u \in U}\zeta_{d_{\gamma}}^{\widetilde{\gamma}q^{u}}\right)X^{\frac{\tau_{\gamma}}{\omega_{\gamma}}\cdot i},$$
	where $d_{\gamma} = \dfrac{\frac{n\omega_{\gamma}}{\tau_{\gamma}}}{\mathrm{gcd}(\frac{n\omega_{\gamma}}{\tau_{\gamma}},\gamma)}$ and $\widetilde{\gamma} = \dfrac{\gamma}{\mathrm{gcd}(\frac{n\omega_{\gamma}}{\tau_{\gamma}},\gamma)}$.
\end{lemma}

Now we give the factorizations of $X^{n}-1$ and of $\Phi_{n}(X)$ over $\mathbb{F}_{q}$. We consider the cases where the module $n$ is odd and where $n$ is even separately. First assume that $n= p_{1}^{e_{1}}\cdots p_{s}^{e_{s}}$ is odd, where $p_{1},\cdots,p_{s}$ are distinct odd primes different from $p$ and $e_{1},\cdots,e_{s}$ are positive integers. Fix a primitive root system $\eta= (\eta_{1},\cdots,\eta_{s})$ modulo $n$. Recall that the set $\Sigma(n;q)$ consists of the $2s$-tuples $(y,x) = (y_{1},\cdots,y_{s},x_{1},\cdots,x_{s})$ of integers satisfying that 
\begin{itemize}
	\item[(1)] $0 \leq y_{i} \leq e_{i}$, $i = 1,\cdots,s$; and
	\item[(2)] $0 \leq x_{i} \leq g_{i,y_{i}}\cdot \mathrm{gcd}(\mathrm{lcm}(f_{1,y_{1}},\cdots,f_{i-1,y_{i-1}}),f_{i,y_{i}}) - 1$, $i = 1,\cdots,s$.
\end{itemize}
Further, define $\Sigma^{\ast}(n;q)$ to be the set of the $s$-tuples $x=(x_{1},\cdots,x_{s})$ of integers such that 
$$0 \leq x_{i} \leq g_{i,0}\cdot \mathrm{gcd}(\mathrm{lcm}(f_{1,0},\cdots,f_{i-1,0}),f_{i,0}) - 1, \ i=1,\cdots,s.$$
Then $\Sigma^{\ast}(n;q)$ can be seen as a subset of $\Sigma(n;q)$ via the embedding 
$$x \mapsto (0,x) = (0,\cdots,0,x_{1},\cdots,x_{s}),$$
and the set
$$\mathcal{CR}_{n/q}^{\ast} = \{\eta_{1}^{x_{1}}\cdots\eta_{s}^{x_{s}} \ | \ x=(x_{1},\cdots,x_{s}) \in \Sigma^{\ast}(n;q)\}$$
forms a full set of representatives of the $q$-cyclotomic cosets modulo $n$ that contain elemnets coprime to $n$.

\begin{theorem}\label{thm 2}
	Let $n= p_{1}^{e_{1}}\cdots p_{s}^{e_{s}}$, where $p_{1},\cdots,p_{s}$ are distinct odd primes different from $p$ and $e_{1},\cdots,e_{s}$ are positive integers. Then the irreducible factorization of $X^{n}-1$ over $\mathbb{F}_{q}$ is given by
	$$X^{n}-1 = \prod_{(y,x) \in \Sigma(n;q)}\left(\sum_{i=0}^{\omega_{y}}(-1)^{\omega_{y}-i}(\sum_{\substack{U\subseteq \{0,\cdots,\omega_{y}-1\}\\ |U|=\omega_{y}-i}}\prod_{u \in U}\zeta_{p_{1}^{e_{1}-y_{1}-M_{y_{1}}}\cdots p_{s}^{e_{s}-y_{s}-M_{y_{s}}}}^{\eta_{1}^{x_{1}}\cdots\eta_{s}^{x_{s}}q^{u}})X^{p_{1}^{M_{y_{1}}}\cdots p_{1}^{M_{y_{s}}}\cdot i}\right),$$
	where for $j = 1.\cdots,s$, $m_{y_{j}} = \mathrm{min}\{e_{j}-y_{j},1\}$, $\omega_{y} = \mathrm{lcm}(\mathrm{ord}_{p_{1}^{m_{y_{1}}}}(q),\cdots,\mathrm{ord}_{p_{s}^{m_{y_{s}}}}(q))$ and $M_{y_{j}} = \mathrm{max}\{e_{j}-y_{j}-v_{p_{j}}(q^{\omega_{y}}-1),0\}$. And  the irreducible factorization of $\Phi_{n}(X)$ over $\mathbb{F}_{q}$ is given by
	$$\Phi_{n}(X) = \prod_{x \in \Sigma^{\ast}(n;q)}\left(\sum_{i=0}^{\omega_{0}}(-1)^{\omega_{0}-i}(\sum_{\substack{U\subseteq \{0,\cdots,\omega_{0}-1\}\\ |U|=\omega_{0}-i}}\prod_{u \in U}\zeta_{p_{1}^{e_{1}-M_{1}}\cdots p_{s}^{e_{s}-M_{s}}}^{\eta_{1}^{x_{1}}\cdots\eta_{s}^{x_{s}}q^{u}})X^{p_{1}^{M_{{1}}}\cdots p_{1}^{M_{{s}}}\cdot i}\right),$$
	where $\omega_{0} = \mathrm{lcm}(\mathrm{ord}_{p_{1}}(q),\cdots,\mathrm{ord}_{p_{s}}(q))$  and $M_{i} = \mathrm{max}\{e_{i}-v_{p_{i}}(q^{\omega_{0}}-1),0\}$, $i=1,\cdots,s$.
\end{theorem}

\begin{proof}
	Since 
	$$\mathcal{CR}_{n/q} = \{\eta^{(y,x)} \ | \ (y,x) \in \Sigma(n^{\prime};q)\}$$
	forms a full set of representatives of $q$-cyclotomic cosets modulo $n$, one has
	$$X^{n}-1 = \prod_{(y,x) \in \Sigma(n^{\prime};q)} M_{\eta^{(y,x)},q}(X).$$
	For any $(y,x) = (y_{1},\cdots,y_{s},x_{1},\cdots,x_{s}) \in \Sigma(n^{\prime};q)$, as $n_{\eta^{(y,x)}} = p_{1}^{m_{y_{1}}}\cdots p_{s}^{m_{y_{s}}}$, then
	$$\omega_{\eta^{(y,x)}} = \mathrm{ord}_{\mathrm{rad}(\eta^{(y,x)})}(q) = \mathrm{lcm}(\mathrm{ord}_{p_{1}^{m_{y_{1}}}}(q),\cdots,\mathrm{ord}_{p_{s}^{m_{y_{s}}}}(q)) = \omega_{y}.$$
	Recall that $|c_{n/q}(\eta^{(y,x)})| = \tau_{y} = \omega_{y}p_{1}^{M_{y_{1}}}\cdots p_{s}^{M_{y_{s}}}$, then it follows Lemma \ref{lem 6} that 
	$$X^{n}-1 = \prod_{(y,x) \in \Sigma(n;q)}\left(\sum_{i=0}^{\omega_{y}}(-1)^{\omega_{y}-i}(\sum_{\substack{U\subseteq \{0,\cdots,\omega_{y}-1\}\\ |U|=\omega_{y}-i}}\prod_{u \in U}\zeta_{p_{1}^{e_{1}-y_{1}-M_{y_{1}}}\cdots p_{s}^{e_{s}-y_{s}-M_{y_{s}}}}^{\eta_{1}^{x_{1}}\cdots\eta_{s}^{x_{s}}q^{u}})X^{p_{1}^{M_{y_{1}}}\cdots p_{1}^{M_{y_{s}}}\cdot i}\right).$$
	Moreover, as $\mathcal{CR}_{n/q}^{\ast}$ defined above is a full set of representatives of the $q$-cyclotomic cosets modulo $n$ that contains elements coprime to $n$, then it holds that 
	$$\Phi_{n}(X) = \prod_{x \in \Sigma^{\ast}(n;q)}\left(\sum_{i=0}^{\omega_{0}}(-1)^{\omega_{0}-i}(\sum_{\substack{U\subseteq \{0,\cdots,\omega_{0}-1\}\\ |U|=\omega_{0}-i}}\prod_{u \in U}\zeta_{p_{1}^{e_{1}-M_{1}}\cdots p_{s}^{e_{s}-M_{s}}}^{\eta_{1}^{x_{1}}\cdots\eta_{s}^{x_{s}}q^{u}})X^{p_{1}^{M_{{1}}}\cdots p_{1}^{M_{{s}}}\cdot i}\right).$$
\end{proof}

We turn to the case that $n$ is even. Assume that $n = 2^{e_{0}}p_{1}^{e_{1}}\cdots p_{s}^{e_{s}}$ where $p_{1},\cdots,p_{s}$ are distinct odd primes different than $p$ and $e_{0},e_{1},\cdots,e_{s}$ are positive integers, and $n^{\prime} = p_{1}^{e_{1}}\cdots p_{s}^{e_{s}}$ is the maximal odd divisor of $n$. Let $\eta = (\eta_{1},\cdots,\eta_{s})$ be a system of primitive roots modulo $n^{\prime}$, and $\Sigma(n^{\prime};q)$ be the set of applicable tuples with respect to $n^{\prime}$ and $q$. Further, we set
$$\Sigma^{+}(n^{\prime};q) = \{(y,x) \in \Sigma(n^{\prime};q) \ | \ q^{\tau_{y}} \equiv 1 \pmod{4}\}$$
and
$$\Sigma^{-}(n^{\prime};q) = \{(y,x) \in \Sigma(n^{\prime};q) \ | \ q^{\tau_{y}} \equiv 3 \pmod{4}\}.$$
The notations $s_{y,d}^{+}$, $s_{y,d}^{-}$, $\mu(y,x)_{d,t}^{+}$ and $\mu(y,x)_{d,t}^{-}$ are defined as in Section \ref{sec 8}.

\begin{theorem}\label{thm 1}
	Let $n = 2^{e_{0}}p_{1}^{e_{1}}\cdots p_{s}^{e_{s}}$, where $p_{1},\cdots,p_{s}$ are distinct odd primes different than $p$ and $e_{0},e_{1},\cdots,e_{s}$ are positive integers. Then the irreducible factorization of $X^{n}-1$ over $\mathbb{F}_{q}$ is given by
	\begin{align*}
		X^{n}-1 = & \prod_{(y,x) \in \Sigma^{+}(n^{\prime};q)}\prod_{d=0}^{e_{0}} \prod_{t \in \{0,1\}^{s_{y,d}^{+}}}\left(\sum_{i=0}^{\omega_{y}}(-1)^{\omega_{y}-i}(\sum_{\substack{U\subseteq \{0,\cdots,\omega_{y}-1\}\\ |U|=\omega_{y}-i}}\prod_{u \in U}\zeta_{r(y,x)_{d,t}^{+}}^{\widetilde{\mu}(y,x)_{d,t}^{+}q^{u}})X^{h(y,x)_{d,t}^{+}\cdot i}\right)\\
		&\times \prod_{(y,x) \in \Sigma^{-}(n^{\prime};q)}\prod_{d=0}^{e_{0}} \prod_{t \in \{0,1\}^{s_{y,d}^{-}}}\left(\sum_{i=0}^{\omega_{y}^{\prime}}(-1)^{\omega_{y}^{\prime}-i}(\sum_{\substack{U\subseteq \{0,\cdots,\omega_{y}^{\prime}-1\}\\ |U|=\omega_{y}^{\prime}-i}}\prod_{u \in U}\zeta_{r(y,x)_{d,t}^{-}}^{\widetilde{\mu}(y,x)_{d,t}^{-}q^{u}})X^{h(y,x)_{d,t}^{-}\cdot i}\right)
	\end{align*}
	where
	\begin{description}
		\item[(1)] for any $(y,x) \in \Sigma^{+}(n^{\prime};q)$,
		$$\widetilde{\mu}(y,x)_{d,t}^{+} = \dfrac{1}{2^{d}}(\eta_{1}^{x_{1}}\cdots\eta_{s}^{x_{s}}+p_{1}^{e_{1}-y_{1}}\cdots p_{s}^{e_{s}-y_{s}}\cdot(U_{d}(\eta^{(y,x)})+\sum_{j=1}^{s_{y,d}^{+}}t_{j}\cdot 2^{d+j})),$$
		$$r(y,x)_{d,t}^{+} = 2^{\mathrm{min}\{e_{0}-d,v^{+}(y)\}}p_{1}^{e_{1}-y_{1}-M_{y_{1}}}\cdots p_{s}^{e_{s}-y_{s}-M_{y_{s}}},$$
		and
		$$h(y,x)_{d,t}^{+} = 2^{\mathrm{max}\{e_{0}-d-v^{+}(y),1\}}p_{1}^{M_{y_{1}}}\cdots p_{s}^{M_{y_{s}}}.$$
		\item[(2)] for any $(y,x) \in \Sigma^{-}(n^{\prime};q)$,
		\begin{equation*}
			\omega_{y}^{\prime} = \left\{
			\begin{array}{lcl}
				2\omega_{y}, \quad \mathrm{if} \ 0 \leq d \leq e_{0}-3;\\
				\omega_{y}, \quad \mathrm{if} \ e_{0}-2 \leq d \leq e_{0},
			\end{array} \right.
		\end{equation*}
		$$\widetilde{\mu}(y,x)_{d,t}^{-} = \dfrac{1}{2^{d}}(\eta_{1}^{x_{1}}\cdots\eta_{s}^{x_{s}}+p_{1}^{e_{1}-y_{1}}\cdots p_{s}^{e_{s}-y_{s}}\cdot(U_{d}(\eta^{(y,x)})+\sum_{j=1}^{s_{y,d}^{-}}t_{j}\cdot 2^{d+j+1})),$$
		\begin{equation*}
			r(y,x)_{d,t}^{-} = \left\{
			\begin{array}{lcl}
				2^{\mathrm{min}\{e_{0}-d,v^{-}(y)+1\}}p_{1}^{e_{1}-y_{1}-M_{y_{1}}}\cdots p_{s}^{e_{s}-y_{s}-M_{y_{s}}}, \quad \mathrm{if} \ 0 \leq d \leq e_{0}-3;\\
				2p_{1}^{e_{1}-y_{1}-M_{y_{1}}}\cdots p_{s}^{e_{s}-y_{s}-M_{y_{s}}}, \quad \mathrm{if} \ d= e_{0}-2 \ \mathrm{or} \ e_{0}-1;\\
				p_{1}^{e_{1}-y_{1}-M_{y_{1}}}\cdots p_{s}^{e_{s}-y_{s}-M_{y_{s}}}, \quad \mathrm{if} \ d= e_{0},
			\end{array} \right.
		\end{equation*}
		and
		\begin{equation*}
			h(y,x)_{d,t}^{-} = \left\{
			\begin{array}{lcl}
				2^{\mathrm{max}\{e_{0}-d-v^{-}(y)-1,0\}}p_{1}^{M_{y_{1}}}\cdots p_{s}^{M_{y_{s}}}, \quad \mathrm{if} \ 0 \leq d \leq e_{0}-3;\\
				2p_{1}^{M_{y_{1}}}\cdots p_{s}^{M_{y_{s}}}, \quad \mathrm{if} \ d= e_{0}-2;\\
				p_{1}^{M_{y_{1}}}\cdots p_{s}^{M_{y_{s}}}, \quad \mathrm{if} \ d= e_{0}-1 \ \mathrm{or} \ e_{0}.
			\end{array} \right.
		\end{equation*}
	\end{description}
\end{theorem}

\begin{proof}
	By Theorem \ref{thm 7} a full set of representatives of $\mathcal{C}_{n/q}$ is
	$$\mathcal{CR}_{n/q} = (\bigsqcup_{(y,x) \in \Sigma^{+}(n^{\prime};q)} \mathcal{CR}_{(y,x)}^{+}) \sqcup (\bigsqcup_{(y,x) \in \Sigma^{-}(n^{\prime};q)} \mathcal{CR}_{(y,x)}^{-}).$$
	By the definition of $\mathcal{CR}_{(y,x)}^{+}$ and $\mathcal{CR}_{(y,x)}^{-}$ one has that 
	$$X^{n}-1 = \prod_{(y,x) \in \Sigma^{+}(n^{\prime};q)}\prod_{d=0}^{e_{0}} \prod_{t \in \{0,1\}^{s_{y,d}^{+}}}M_{\mu(y,x)_{d,t}^{+},q}(X) \times \prod_{(y,x) \in \Sigma^{-}(n^{\prime};q)}\prod_{d=0}^{e_{0}} \prod_{t \in \{0,1\}^{s_{y,d}^{-}}}M_{\mu(y,x)_{d,t}^{-},q}(X).$$

	For any $(y,x) \in \Sigma^{+}(n^{\prime};q)$, $d \in \{0,\cdots,e_{0}\}$ and $t \in \{0,1\}^{s_{y,d}^{+}}$, as $q^{\tau_{y}} \equiv 1 \pmod{4}$, Lemma \ref{lem 6} yields that 
	$$M_{\mu(y,x)_{d,t}^{-},q}(X) = \sum_{i=0}^{\omega_{y}-1}(-1)^{\omega_{y}-i}(\sum_{\substack{U \subseteq \{0,\cdots,\omega_{y}-1\}\\ |U|=\omega_{y}-i}}\prod_{u \in U}\zeta_{r(y,x)_{d,t}^{+}}^{\widetilde{\mu}(y,x)_{d,t}^{+}q^{u}})X^{h(y,x)_{d,t}^{+}\cdot i},$$
	where $r(y,x)_{d,t}^{+} = \dfrac{\frac{n\omega_{y}}{\tau}}{\mathrm{gcd}(\frac{n\omega_{y}}{\tau},\mu(y,x)_{d,t}^{+})}$, and $\widetilde{\mu}(y,x)_{d,t}^{+} = \dfrac{\mu(y,x)_{d,t}^{+}}{\mathrm{gcd}(\frac{n\omega_{y}}{\tau},\mu(y,x)_{d,t}^{+})}$ and $h(y,x)_{d,t}^{+} = \frac{\tau}{\omega_{y}}$ for $\tau = |c_{n/q}(\mu(y,x)_{d,t}^{+})|$. Now the assertions follow from Corollary \ref{coro 3}. Similar argument applies for the case where $(y,x) \in \Sigma^{-}(n^{\prime};q)$, $d \in \{0,\cdots,e_{0}\}$ and $t \in \{0,1\}^{s_{y,d}^{-}}$.
\end{proof}

Notice that for any representative $\mu(y,x)_{d,t}^{\pm}$, it is coprime to $n$ if and only if $(y,x) = (0,x) \in \Sigma^{\ast}(n^{\prime};q)$ and $d=0$. Then from Theorem \ref{thm 1} we obtain the following corollary.

\begin{corollary}
	If $q^{\mathrm{ord}_{n^{\prime}}(q)} \equiv 1 \pmod{4}$, then the irreducible factorization of $\Phi_{n}(X)$ is given by
	$$\Phi_{n}(X)= \prod_{x \in \Sigma^{\ast}(n^{\prime};q)}\prod_{t \in \{0,1\}^{s_{0,0}^{+}}}\left(\sum_{i=0}^{\omega_{0}}(-1)^{\omega_{0}-i}(\sum_{\substack{U\subseteq \{0,\cdots,\omega_{0}-1\}\\ |U|=\omega_{0}-i}}\prod_{u \in U}\zeta_{r(x)_{t}^{+}}^{\widetilde{\mu}(x)_{t}^{+}q^{u}})X^{h(x)_{t}^{+}\cdot i}\right),$$
	where
	$$\widetilde{\mu}(x)_{t}^{+} = \eta_{1}^{x_{1}}\cdots\eta_{s}^{x_{s}}+p_{1}^{e_{1}}\cdots p_{s}^{e_{s}}\cdot(U_{d}(\eta_{1}^{x_{1}}\cdots\eta_{s}^{x_{s}})+\sum_{j=1}^{s_{0,0}^{+}}t_{j}\cdot 2^{j}),$$
	$$r(x)_{t}^{+} = 2^{\mathrm{min}\{e_{0},v_{2}(q^{\mathrm{ord}_{n^{\prime}}(q)}-1)\}}p_{1}^{e_{1}-M_{1}}\cdots p_{s}^{e_{s}-M_{s}},$$
	and
	$$h(x)_{t}^{+} = 2^{\mathrm{max}\{e_{0}-v_{2}(q^{\mathrm{ord}_{n^{\prime}}(q)}-1),1\}}p_{1}^{M_{1}}\cdots p_{s}^{M_{s}}.$$
	If $q^{\mathrm{ord}_{n^{\prime}}(q)} \equiv 3 \pmod{4}$, then the irreducible factorization of $\Phi_{n}(X)$ is given by
	$$\Phi_{n}(X)= \prod_{x \in \Sigma^{\ast}(n^{\prime};q)}\prod_{t \in \{0,1\}^{s_{0,0}^{-}}}\left(\sum_{i=0}^{\omega_{0}^{\prime}}(-1)^{\omega_{0}^{\prime}-i}(\sum_{\substack{U\subseteq \{0,\cdots,\omega_{0}^{\prime}-1\}\\ |U|=\omega_{0}^{\prime}-i}}\prod_{u \in U}\zeta_{r(x)_{t}^{-}}^{\widetilde{\mu}(x)_{t}^{-}q^{u}})X^{h(x)_{t}^{-}\cdot i}\right),$$
	where
	\begin{equation*}
		\omega_{0}^{\prime} = \left\{
		\begin{array}{lcl}
			\omega_{0}, \quad \mathrm{if} \ e_{0}\leq 2;\\
			2\omega_{0}, \quad \mathrm{if} \ e_{0}\geq 3,
		\end{array} \right.
	\end{equation*}
	$$\widetilde{\mu}(x)_{t}^{-} = \eta_{1}^{x_{1}}\cdots\eta_{s}^{x_{s}}+p_{1}^{e_{1}}\cdots p_{s}^{e_{s}}\cdot(U_{d}(\eta_{1}^{x_{1}}\cdots\eta_{s}^{x_{s}})+\sum_{j=1}^{s_{0,0}^{-}}t_{j}\cdot 2^{j+1}),$$
	\begin{equation*}
		r(x)_{t}^{-} = \left\{
		\begin{array}{lcl}
			2p_{1}^{e_{1}-M_{1}}\cdots p_{s}^{e_{s}-M_{s}}, \quad \mathrm{if} \ e_{0}=1 \ \mathrm{or} \ 2;\\
			2^{\mathrm{min}\{e_{0},v_{2}(q^{\mathrm{ord}_{n^{\prime}}(q)}+1)+1\}}p_{1}^{e_{1}-M_{1}}\cdots p_{s}^{e_{s}-M_{s}}, \quad \mathrm{if} \ e_{0} \geq 3,
		\end{array} \right.
	\end{equation*}
	and
	\begin{equation*}
		h(x)_{t}^{-} = \left\{
		\begin{array}{lcl}
			p_{1}^{M_{1}}\cdots p_{s}^{M_{s}}, \quad \mathrm{if} \ e_{0}=1;\\
			2p_{1}^{M_{1}}\cdots p_{s}^{M_{s}}, \quad \mathrm{if} \ e_{0}=2;\\
			2^{\mathrm{max}\{e_{0}-v_{2}(q^{\mathrm{ord}_{n^{\prime}}(q)}+1)-1,0\}}p_{1}^{M_{1}}\cdots p_{s}^{M_{s}}, \quad \mathrm{if} \ e_{0}=3.
		\end{array} \right.
	\end{equation*}
\end{corollary}

\section{Cyclic codes over finite fields}\label{sec 6}
As an application of the irreducible factorization of $X^{n}-1$ over $\mathbb{F}_{q}$, we determine the cyclic codes over $\mathbb{F}_{q}$ via giving their generator polynomials. To deal with both the simple-rooted and the repeat-rooted cyclic codes, we allow the length $m$ of the codes to be divisible by $p$. Furthermore, we prove a criterion for a cyclic code to be self-dual, and enumerate the self-dual cyclic codes of length $m$ over $\mathbb{F}_{q}$.

Let $m = p^{v}n$ be a positive integer, where $v = v_{p}(m) \geq 0$ and $n$ is coprime to $p$. Following from Theorem \ref{thm 2} and \ref{thm 1} directly, the cyclic codes of length $m$ over $\mathbb{F}_{q}$ are given by the following theorems, in the cases where $n$ is odd and where $n$ is even respectively.

\begin{theorem}\label{thm 4}
	Assume that $n= p_{1}^{e_{1}}\cdots p_{s}^{e_{s}}$, where $p_{1},\cdots,p_{s}$ are distinct odd primes different from $p$ and $e_{1},\cdots,e_{s}$ are positive integers. With the notations adopted from Theorem \ref{thm 2}, all the cyclic codes of length $m = p^{v}n$ are given by
	$$C_{\Omega} = \left(\prod_{(y,x) \in \Sigma(n;q)}M_{\eta^{(y,x)},q}(X)^{\Omega(y,x)}\right)$$
	for all $0 \leq \Omega(y,x) \leq p^{v}$, $(y,x) \in \Sigma(n;q)$, where 
	$$M_{\eta^{(y,x)},q}(X) = \sum_{i=0}^{\omega_{y}}(-1)^{\omega_{y}-i}(\sum_{\substack{U\subseteq \{0,\cdots,\omega_{y}-1\}\\ |U|=\omega_{y}-i}}\prod_{u \in U}\zeta_{p_{1}^{e_{1}-y_{1}-M_{y_{1}}}\cdots p_{s}^{e_{s}-y_{s}-M_{y_{s}}}}^{\eta_{1}^{x_{1}}\cdots\eta_{s}^{x_{s}}q^{u}})X^{p_{1}^{M_{y_{1}}}\cdots p_{1}^{M_{y_{s}}}\cdot i}.$$
\end{theorem}

\begin{theorem}
	Assume that $n = 2^{e_{0}}p_{1}^{e_{1}}\cdots p_{s}^{e_{s}}$, where $p_{1},\cdots,p_{s}$ are distinct odd primes different than $p$ and $e_{0},e_{1},\cdots,e_{s}$ are positive integers. With the notations adopted from Theorem \ref{thm 1}, all the cyclic codes of length $m = p^{v}n$ are given by
	$$C_{\Omega} = \left(\prod_{(y,x) \in \Sigma^{+}(n^{\prime};q)}\prod_{d=0}^{e_{0}} \prod_{t \in \{0,1\}^{s_{y,d}^{+}}}M_{\mu(y,x)_{d,t}^{+},q}(X)^{\Omega(y,x)_{d,t}^{+}}\times\prod_{(y,x) \in \Sigma^{-}(n^{\prime};q)}\prod_{d=0}^{e_{0}} \prod_{t \in \{0,1\}^{s_{y,d}^{-}}}M_{\mu(y,x)_{d,t}^{-},q}(X)^{\Omega(y,x)_{d,t}^{-}}\right)$$
	for all $0 \leq \Omega(y,x)_{d,t}^{+}, \Omega(y,x)_{d,t}^{-} \leq p^{v}$, where 
	$$M_{\mu(y,x)_{d,t}^{+},q}(X) = \sum_{i=0}^{\omega_{y}}(-1)^{\omega_{y}-i}(\sum_{\substack{U\subseteq \{0,\cdots,\omega_{y}-1\}\\ |U|=\omega_{y}-i}}\prod_{u \in U}\zeta_{r(y,x)_{d,t}^{+}}^{\widetilde{\mu}(y,x)_{d,t}^{+}q^{u}})X^{t(y,x)_{d,t}^{+}\cdot i}$$
	and 
	$$M_{\mu(y,x)_{d,t}^{-},q}(X) = \sum_{i=0}^{\omega_{y}^{\prime}}(-1)^{\omega_{y}^{\prime}-i}(\sum_{\substack{U\subseteq \{0,\cdots,\omega_{y}^{\prime}-1\}\\ |U|=\omega_{y}^{\prime}-i}}\prod_{u \in U}\zeta_{r(y,x)_{d,t}^{-}}^{\widetilde{\mu}(y,x)_{d,t}^{-}q^{u}})X^{t(y,x)_{d,t}^{-}\cdot i}.$$
\end{theorem}

As it is well known that self-dual cyclic codes of length $m$ over $\mathbb{F}_{q}$ exist if and only if $q$ is a power of $2$ and $m$ is even, in the following context we restrict ourselves to this case. Assume that $q=2^{e}$, and $m = 2^{v}n$ where $p_{1},\cdots,p_{s}$ are distinct odd primes and $v,e_{0},e_{1},\cdots,e_{s}$ are positive integers. By theorem \ref{thm 4} any cyclic code of length $m$ over $\mathbb{F}_{2^{e}}$ has the generator polynomial in the form
$$\prod_{\Sigma(n;2^{e})}M_{\eta^{(y,x)},2^{e}}(X)^{\Omega(y,x)},$$
where $0 \leq \Omega(y,x) \leq 2^{v}$.

\begin{lemma}\label{lem 5}
	Let $(y,x) = (y_{1},\cdots,y_{s},x_{1},\cdots,x_{s}) \in \Sigma(n;2^{e})$, and let $\{i_{1},\cdots,i_{r}\}$ be the subset of $\{1,\cdots,s\}$ consisting of the index $i$ such that $y_{i} < e_{i}$. The polynomial $M_{\eta^{(y,x)},2^{e}}(X)$ is self-dual if and only if 
	$$v_{2}(\mathrm{ord}_{p_{i_{1}}}(q)) = \cdots = v_{2}(\mathrm{ord}_{p_{i_{r}}}(q)) > 0.$$
\end{lemma}

\begin{proof}
	The polynomial $M_{\eta^{(y,x)},2^{e}}(X)$ is self-reciprocal if and only if $c_{n/q}(\eta^{(y,x)}) = c_{n/q}(-\eta^{(y,x)})$, that is, there is a positive integer $j$ such that 
	$$-\eta^{(y,x)} \equiv \eta_{1}^{x_{1}}\cdots \eta^{(y,x)}q^{j} \pmod{p_{1}^{e_{1}}\cdots p_{s}^{e_{s}}},$$
	which by the Chinese remainder theorem amounts to 
	\begin{equation}\label{eq 6}
		\left\{
		\begin{array}{lcl}
			q^{j} \equiv -1 \pmod{p_{1}^{e_{1}-y_{1}}}\\
			\cdots \cdots\\
			q^{j} \equiv -1 \pmod{p_{s}^{e_{s}-y_{s}}}
		\end{array} \right.
	\end{equation}
	Since the groups $\mathbb{Z}/p_{i_{k}}^{e_{i_{k}}-y_{i_{k}}}\mathbb{Z}$, $k = 1,\cdots,r$, are cyclic, the congruence equations \eqref{eq 6} are equivalent to that the orders $\mathrm{ord}_{p_{i_{k}}}(q)$ are all even, and 
	\begin{equation}\label{eq 7}
		j \equiv \dfrac{1}{2}\mathrm{ord}_{p_{i_{k}}^{e_{i_{k}}-y_{i_{k}}}}(q) \pmod{\mathrm{ord}_{p_{i_{k}}^{e_{i_{k}}-y_{i_{k}}}}(q)}, \ k = 1,\cdots,r.
	\end{equation}
	The lift-the-exponent lemma indicates
	$$\mathrm{ord}_{p_{i_{k}}^{e_{i_{k}}-y_{i_{k}}}}(q) = \mathrm{ord}_{p_{i_{k}}}(q)\cdot p_{i_{k}}^{\mathrm{max}\{0,e_{i}-y_{i}-v_{p_{i}}(q^{\mathrm{ord}_{p_{i_{k}}}(q)}-1))\}},$$
	therefore the congruence equations \eqref{eq 7} is solvable if and only if 
	$$v_{2}(\mathrm{ord}_{p_{i_{1}}}(q)) = \cdots = v_{2}(\mathrm{ord}_{p_{i_{r}}}(q)) > 0.$$
\end{proof}

Without lossing generality, we assume that the odd primes $p_{1},\cdots,p_{s}$ are ranked in the way so that 
\begin{small}
	$$0 = v_{2}(\mathrm{ord}_{p_{1}}(q)) = \cdots v_{2}(\mathrm{ord}_{p_{s_{1}}}(q)) < v_{2}(\mathrm{ord}_{p_{s_{1}+1}}(q)) = \cdots = v_{2}(\mathrm{ord}_{p_{s_{2}}}(q)) < \cdots < v_{2}(\mathrm{ord}_{p_{s_{t-1}+1}}(q)) = \cdots = v_{2}(\mathrm{ord}_{p_{s_{t}}}(q)),$$
\end{small}
where $0 \leq s_{1}< \cdots <s_{t} = s$. For every $1 \leq k \leq t$, define a subset $\Sigma_{k}$ of $\Sigma(n;q)$ as
$$\Sigma_{k} = \{(y,x) \in \Sigma(n;q) \ | \ y_{i} = e_{i}, \ \forall 1 \leq i \leq s_{k-1} \ \mathrm{and} \ s_{k}+1 \leq i \leq s\}.$$
Further, we denote by $\Sigma^{\prime}(n;q) = \Sigma_{2}\sqcup \cdots \sqcup \Sigma_{t}$ and $\Sigma^{\prime\prime}(n;q) = \Sigma(n;q) \setminus \Sigma^{\prime}(n;q)$. Then Lemma \ref{lem 5} can be rephrased as that the index in $\Sigma^{\prime}(n;q)$ give rise to self-reciprocal irreducible factors of $X^{m}-1$, while the index in $\Sigma^{\prime\prime}(n;q)$ give rise to non-self-reciprocal irreducible factors.

\begin{theorem}\label{thm 8}
	Let the notations be defined as above. A cyclic code 
	$$C_{\Omega} = \left(\prod_{(y,x) \in \Sigma(n;q)}M_{\eta^{(y,x)},q}(X)^{\Omega(y,x)}\right)$$
	of length $m$ over $\mathbb{F}_{q}$ is self-dual if and only if $\Omega(y,x) = 2^{v-1}$ for all $(y,x) \in \Sigma^{\prime}(n;q)$, and $\Omega(y,x) = 2^{v} - \Omega(y^{\prime},x^{\prime})$ for all pairs $(y,x),(y^{\prime},x^{\prime}) \in \Sigma^{\prime\prime}(n;q)$ such that $-\eta^{(y,x)} \equiv \eta^{(y,x)}q^{j} \pmod{n}$ for some $j \in \mathbb{N}$. Consequently, all the self-dual cyclic codes of length $m$ over $\mathbb{F}_{q}$ are given by
	$$C_{\Omega} = \left(\prod_{(y,x) \in \Sigma^{\prime}(n;q)}M_{\eta^{(y,x)},q}(X)^{2^{v-1}} \times \prod_{(y,x) \in \Sigma^{\prime\prime}(n;q)/\sim}M_{\eta^{(y,x)},q}(X)^{\Omega(y,x)}M_{-\eta^{(y,x)},q}(X)^{2^{v}-\Omega(y,x)}\right),$$
	where the equivalence relation $\sim$ on $\Sigma^{\prime\prime}(n;q)$ is defined by $(y,x) \sim (y^{\prime},x^{\prime})$ if and only if $\pm\eta^{(y,x)} \equiv \eta^{(y,x)}q^{j} \pmod{n}$ for some $j \in \mathbb{N}$, and the polynomial $M_{-\eta^{(y,x)},q}(X)$ can be expressed as
	$$M_{-\eta^{(y,x)},q}(X) = \sum_{i=0}^{\omega_{y}}(-1)^{\omega_{y}-i}(\sum_{\substack{U\subseteq \{0,\cdots,\omega_{y}-1\}\\ |U|=\omega_{y}-i}}\prod_{u \in U}\zeta_{p_{1}^{e_{1}-y_{1}-M_{y_{1}}}\cdots p_{s}^{e_{s}-y_{s}-M_{y_{s}}}}^{-\eta_{1}^{x_{1}}\cdots\eta_{s}^{x_{s}}q^{u}})X^{p_{1}^{M_{y_{1}}}\cdots p_{1}^{M_{y_{s}}}\cdot i}.$$
\end{theorem}

\begin{proof}
	It is obtained by combining Theorem \ref{thm 4} and Lemma \ref{lem 5}.
\end{proof}

\begin{corollary}
	\begin{itemize}
		\item[(1)] $|\Sigma(n;q)| = \sum\limits_{0 \leq y_{1} \leq e_{1}}\cdots \sum\limits_{0 \leq y_{s} \leq e_{s}}\dfrac{\phi(p_{1}^{e_{1}-y_{1}})\cdots \phi(p_{s}^{e_{s}-y_{s}})}{\mathrm{lcm}(f_{1,y_{1}},\cdots,f_{s,y_{s}})}$.
		\item[(2)] For any $1\leq k\leq t$, 
		$$|\Sigma_{k}| = \sum_{0 \leq y_{s_{k-1}+1} \leq e_{s_{k-1}+1}}\cdots \sum_{0 \leq y_{s_{k}} \leq e_{s_{k}}}\dfrac{\phi(p_{s_{k-1}+1}^{e_{s_{k-1}+1}-y_{s_{k-1}+1}})\cdots \phi(p_{s_{k}}^{e_{s_{k}}-y_{s_{k}}})}{\mathrm{lcm}(f_{s_{k-1}+1,y_{s_{k-1}+1}},\cdots,f_{s_{k},y_{s_{k}}})}.$$
		\item[(3)] $|\Sigma^{\prime}(n;q)| = \sum\limits_{k=1}^{t}|\Sigma_{k}|$, and $|\Sigma^{\prime\prime}(n;q)|= |\Sigma(n;q)| - |\Sigma^{\prime}(n;q)|$.
		\item[(4)] There are in total $(2^{v}+1)^{\frac{|\Sigma^{\prime\prime}(n;q)|}{2}}$ self-dual cyclic codes of length $n$ over the finite field $\mathbb{F}_{q}$.
	\end{itemize}
\end{corollary}

\begin{proof}
	For $(1)$, there are in total 
	\begin{align*}
		&\sum_{0 \leq y_{1} \leq e_{1}}\cdots \sum_{0 \leq y_{s} \leq e_{s}} \dfrac{\phi(p_{1}^{e_{1}-y_{1}})}{f_{1,y_{1}}} \dfrac{\phi(p_{2}^{e_{2}-y_{2}})}{f_{2,y_{2}}} \mathrm{gcd}(f_{1,y_{1}},f_{2,y_{2}}) \cdots  \dfrac{\phi(p_{s}^{e_{s}-y_{s}})}{f_{s,y_{s}}} \mathrm{gcd}(\mathrm{lcm}(f_{1,y_{1}},\cdots,f_{s-1,y_{s-1}}),f_{s,y_{s}}) \\
		&= \sum_{0 \leq y_{1} \leq e_{1}}\cdots \sum_{0 \leq y_{s} \leq e_{s}}\dfrac{\phi(p_{1}^{e_{1}-y_{1}})\cdots \phi(p_{s}^{e_{s}-y_{s}})}{\mathrm{lcm}(f_{1,y_{1}},\cdots,f_{s,y_{s}})}
	\end{align*}
	$q$-cyclotomic cosets modulo $n$. That is the order of $\Sigma(n;q)$.
	
	The proof of Conclusion $(2)$ is similar, and $(3)$ is clear. Finally, $(4)$ follows from Theorem \ref{thm 8}.
\end{proof}

\section*{Acknowledgment}
This work was supported by Natural Science Foundation of Beijing Municipal(M23017). 

\section*{Data availability}
Data sharing not applicable to this article as no datasets were generated or analysed during the current study.

\section*{Declaration of competing interest}
The authors declare that we have no known competing financial interests or personal relationships that 
could have appeared to influence the work reported in this paper.

\end{document}